\documentclass[onecolumn,aps,nofootinbib,superscriptaddress,tightenlines,11pt]{revtex4}
\usepackage{amsmath}
\usepackage{amsthm}
\usepackage{amsfonts}
\usepackage{color}
\usepackage{mathtools}
\usepackage{overpic}
\usepackage{subfigure}
\usepackage{verbatim}

\usepackage[ colorlinks = true,
             linkcolor = blue,
             urlcolor  = blue,
             citecolor = red,
             anchorcolor = green,
]{hyperref}

\newtheorem{lemma}{Lemma}
\newtheorem{result}{Result}

\newtheorem{theorem}[lemma]{Theorem}
\newtheorem{corollary}[lemma]{Corollary}

\newcommand{\cS}{\mathcal{S}}
\newcommand{\cL}{\mathcal{L}}
\newcommand{\cP}{\mathcal{P}}
\newcommand{\cH}{\mathcal{H}}
\newcommand{\cE}{\mathcal{E}}
\newcommand{\cD}{\mathcal{D}}
\newcommand{\cN}{\mathcal{N}}
\newcommand{\cM}{\mathcal{M}}
\newcommand{\cZ}{\mathcal{Z}}  

\newcommand{\proj}[1]{\left|#1\middle\rangle\!\middle\langle #1\right|}

\newcommand{\ket}[1]{\left|#1\right\rangle}

\newcommand{\eps}{\varepsilon}

\DeclareMathOperator{\tr}{tr}

\begin{document}

\title{Quantum Coding with Finite Resources}

\author{Marco Tomamichel}
\affiliation{School of Physics, University of Sydney, Sydney, NSW 2006, Australia}
\author{Mario Berta}
\affiliation{Institute for Quantum Information and Matter, Caltech, Pasadena, CA 91125, USA}
\author{Joseph M.~Renes}
\affiliation{Institute for Theoretical Physics, ETH Zurich, 8093 Z\"urich, Switzerland}

\begin{abstract}
The quantum capacity of a memoryless channel is often used as a single figure of merit to characterize its ability to transmit quantum information coherently. The capacity determines the maximal rate at which we can code reliably over asymptotically many uses of the channel. We argue that this asymptotic treatment is insufficient to the point of being irrelevant in the quantum setting where decoherence severely limits our ability to manipulate large quantum systems in the encoder and decoder. For all practical purposes we should instead focus on the trade-off between three parameters: the rate of the code, the number of coherent uses of the channel, and the fidelity of the transmission. The aim is then to specify the region determined by allowed combinations of these parameters. 

Towards this goal, we find approximate and exact characterizations of the region of allowed triplets for the qubit dephasing channel and for the erasure channel with classical post-processing assistance. In each case the region is parametrized by a second channel parameter, the quantum channel dispersion.
In the process we also develop several general inner (achievable) and outer (converse) bounds on the coding region that are valid for all finite-dimensional quantum channels and can be computed efficiently. Applied to the depolarizing channel, this allows
us to determine a lower bound on the number of coherent uses of the channel necessary to witness super-additivity of the coherent information.
\end{abstract}

\maketitle


\section{Introduction}\label{sec:intro}

One of the quintessential topics in quantum information theory is the study of reliable quantum information transmission over a noisy quantum channel. Here the word ``channel'' simply refers to a description of a physical evolution (by means of a completely positive trace-preserving map on density operators). Traditionally one considers point-to-point communication settings where a memoryless channel can be used many times in sequence. The sender (often called Alice) first encodes a quantum state into a sequence of registers and then sends them one by one trough the channel to the receiver (often called Bob). Bob collects these registers and then attempts to decode the quantum state. Alternatively, consider a collection of physical qubits that are exposed to independent noise. The goal is then to encode quantum information (logical qubits) into this system so that the quantum information can be decoded with high fidelity at a later stage. 
One of the primary goals of information theory is to find fundamental limits imposed on any coding scheme that tries to accomplish such tasks.

Following a tradition going back to Shannon's groundbreaking work~\cite{shannon48}, this problem is usually studied asymptotically: the \emph{quantum capacity} of a channel is defined as the \emph{optimal rate} (in qubits per use of the channel) at which we can transmit quantum information \emph{with vanishing error} as the number of channel uses \emph{goes to infinity}. In the context of information storage, the rate simply corresponds to the ratio of logical to physical qubits, and the number of channel uses corresponds to the number of physical qubits. The quantum capacity of arbitrary channels has been determined in a series of works, an upper bound to the capacity shown in~\cite{schumacher_quantum_1996,barnum_information_1998,barnum00} and achievability of that bound shown in~\cite{lloyd97,shor02,devetak05b}. 

However, in any application of the theory resources are \emph{finite} and the number of channel uses is necessarily limited. More importantly, at least for the near future it appears unrealistic to expect that encoding and decoding circuits can coherently manipulate large numbers of qubits. Restricting the size of the quantum devices used for encoding the channel inputs and decoding its outputs is tantamount to considering communication with only a fixed number of channel uses. This then raises the question whether an asymptotic approach where this number goes to infinity\,---\,which has proven to be very successful for the analysis of classical communication systems\,---\,is equally suitable for the quantum setting.
Clearly, what we really want to understand is how well we can transmit quantum information in a realistic setting where the number of channel uses and the size of quantum devices is limited.
The quantum capacity is at most a proxy for the answer to this question, and in this article we argue that it is often not a very good one. 

The study of such non-asymptotic scenarios has recently garnered significant attention in classical information theory~\cite{polyanskiy10,hayashi09,tanbook14} as well as in quantum information theory~\cite{tomamichel12,li12,tomamicheltan14,datta14}. Here we extend these considerations to the setting of quantum communication.


\paragraph*{Outline:}

The remainder of the paper is structured as follows. In Section~\ref{sec:results} we discuss our main results detail. In Section~\ref{sec:notation} we formally introduce relevant notation and information measures. In Section~\ref{sec:converse} we derive our converse (outer) bounds and in Section~\ref{sec:achieve} we derive our achievability (inner) bounds. Finally, Section~\ref{sec:examples} discusses the specific examples presented below as Results~\ref{res:z},~\ref{res:e}, and~\ref{res:depol}.


\section{Discussion of Results}\label{sec:results}

\begin{figure}
\begin{overpic}[width=.5\textwidth]{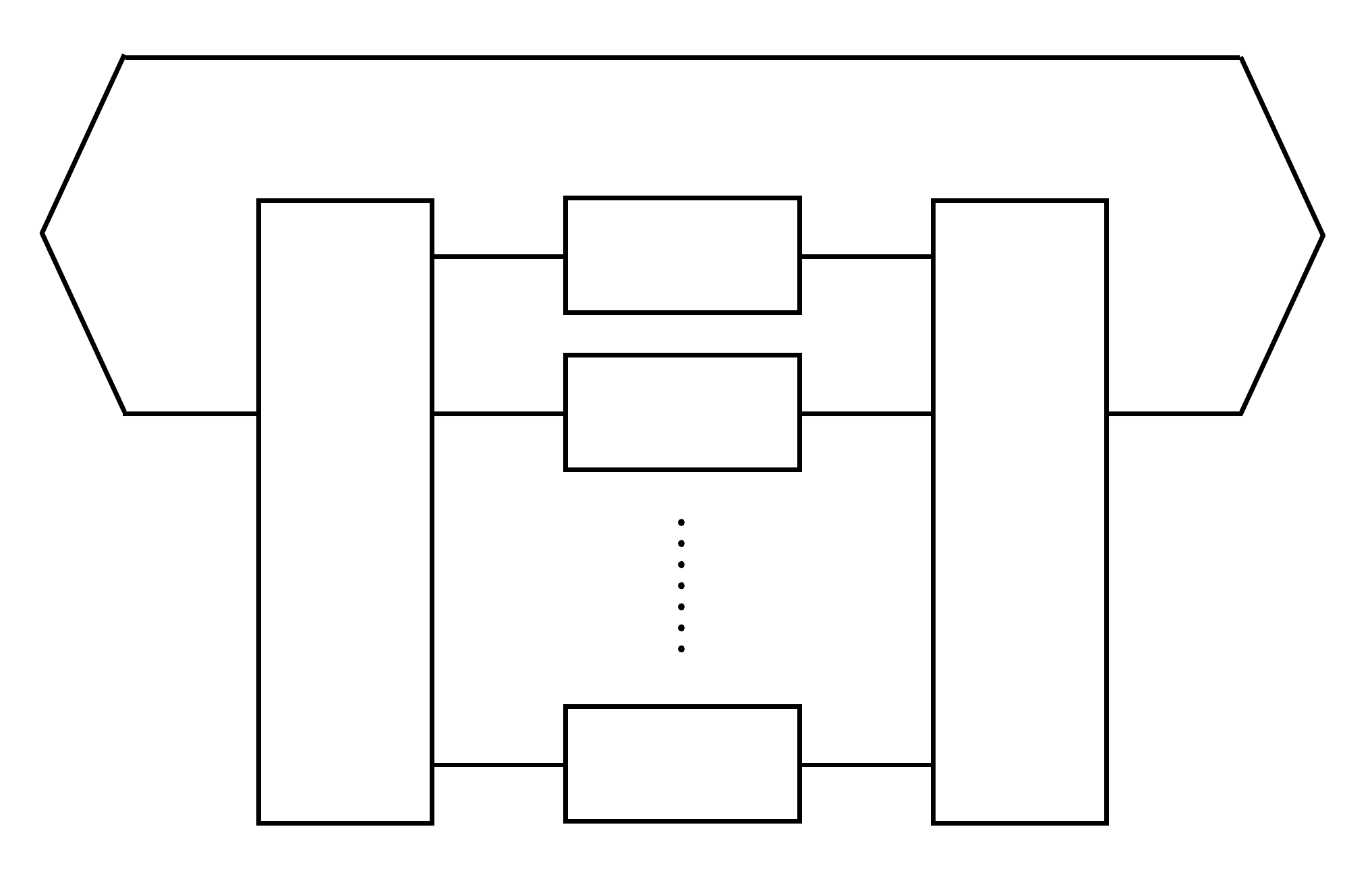}
  \put(24,26){\footnotesize$\cE$}
  \put(73,26){\footnotesize$\cD$}
  \put(48,8){\footnotesize$\cN$}
  \put(48,33){\footnotesize$\cN$}
  \put(48,45){\footnotesize$\cN$}
  \put(-9,48){\footnotesize${\phi_{MM'}}$}
  \put(98,48){\footnotesize${\rho_{MM''}}$}
  \put(11,36){\footnotesize$M'$}
  \put(83,36){\footnotesize$M''$}
  \put(11,56.5){\footnotesize$M$}
  \put(34.5,10.5){\footnotesize$A_n$}
  \put(34.5,36){\footnotesize$A_2$}
  \put(34.5,47.5){\footnotesize$A_1$}
  \put(60,10.5){\footnotesize$B_n$}
  \put(60,36){\footnotesize$B_2$}
  \put(60,47.5){\footnotesize$B_1$}
\end{overpic}
  \caption{Coding Scheme for entanglement transmission over $n$ uses of a channel $\cN \equiv \cN_{A\to B}$. The systems $M$, $M'$ and $M''$ are isomorphic. The encoder $\cE \equiv \cE_{M' \to A^n}$ encodes the part $M'$ of the maximally entangled state $\phi_{MM'}$ into the channel input systems. Later, the decoder $\cD \equiv \cD_{B^n \to M''}$ recovers the state from the channel output systems.
  The performance of the code is measured using the fidelity $F(\phi_{MM'}, \rho_{MM''})$.}
  \label{fig:scheme}
\end{figure}

In this work we focus on codes enabling a state entangled with a reference system to be reliably transmitted through the channel. This is a strong requirement: reliable entanglement transmission implies reliable transmission, on average, of all non-entangled input states. The coding scheme is depicted in Figure~\ref{fig:scheme}. We are given a quantum channel $\cN \equiv \cN_{A\to B}$ and denote by $\cN^{\otimes n}$ the $n$-fold parallel repetition of this channel. An \emph{entanglement transmission code} for $\cN^{\otimes n}$ is given by a triplet $\{ |M|, \cE, \cD\}$, where $|M|$ is the local dimension of a maximally entangled state $\ket{\phi}_{MM'}$ that is to be transmitted over $\cN^{\otimes n}$. The quantum channels $\cE \equiv \cE_{M' \to A^n}$ and $\cD \equiv \cD_{B^n \to M''}$ are encoding and decoding operations, respectively. (A more formal treatment will follow in Section~\ref{sec:notation-codes}.) With this in hand, we now say that a triplet $\{R, n, \eps\}$ is \emph{achievable} on the channel $\cN$ if there exists an entanglement transmission code satisfying
\begin{align}\label{eq:coding_def}
    \frac{1}{n} \log |M| \geq R  \qquad \textrm{and} \qquad F\Big( \phi_{MM'},\, (\cD \circ \cN^{\otimes n} \circ\cE)(\phi_{MM'}) \Big) \geq 1-\eps\,.
\end{align}
Here, $R$ is the \emph{rate} of the code, $n$ is the number of channel uses, and $\eps$ is the tolerated error measured in terms of the fidelity $F$. 

The non-asymptotic \emph{achievable region} of a quantum channel $\cN$ is then given by the union of all achievable triplets $\{ R, n, \eps\}$. The goal of (non-asymptotic) information theory is to find tight bounds on this achievable region, in particular to determine if certain triplets are outside the achievable region and thus \emph{forbidden}. For this purpose, we define its boundary 
\begin{align}\label{eq:optrate}
\hat{R}_{\cN}(n; \eps) := \max \big\{ R : (R,n,\eps) \textrm{ is achievable on } \cN \big\}\,,
\end{align}
and investigate it as a function of $n$ for a fixed value of $\eps$.\footnote{An alternative approach would be to investigate the boundary $\hat{\eps}_{\cN}(n; R) := \max \{ \eps : (R, n, \eps) \textrm{ is achievable} \}$. This leads to the study of error exponents (and the reliability function) as well as strong converse exponents. We will not discuss this here since such an analysis usually does not yield sufficiently tight bounds for small values of $n$.} We will often drop the subscript $\cN$ if it is clear which channel is considered.

To begin, let us rephrase the seminal capacity results~\cite{schumacher_quantum_1996,barnum_information_1998,barnum00,lloyd97,shor02,devetak05} in this language. The quantum capacity is defined as the asymptotic limit of $\hat{R}_{\cN}(n;\eps)$ when $n$ (first) goes to infinity and $\eps$ vanishes. 
The capacity can be expressed in terms of a regularized coherent information:
\begin{align}\label{eq:cap}
   Q(\cN):= \lim_{\eps \to 0} \lim_{n \to \infty} \hat{R}_{\cN}(n;\eps)= \lim_{\ell \to \infty} \frac{I_{\textrm{c}}\big(\cN^{\otimes \ell}\big)}{\ell}\,,
\end{align}
where the coherent information $I_{\textrm{c}}$ will be defined in Section~\ref{sec:achieve}. This result is highly unsatisfactory, not least because the regularization makes its computation intractable.\footnote{It is not clear if the limit $\ell \to \infty$ is necessary for any fixed channel, but it was recently shown that there does not exist a universal constant $\ell_0$ such that $C(\cN) \leq \frac{1}{\ell_0} I_{\textrm{c}}(\cN^{\ell_0})$ for all channels $\cN$~\cite{cubitt14}.} Worse, the statement is not as strong as we would like it to be because it does not give any indication of the fundamental limits for finite $\eps$ or finite $n$.

For example, even sticking to the asymptotic limit for now, we might be willing to admit a small but nonzero error in our recovery. Formally, instead of requiring that the error vanishes asymptotically, we only require that
it does not exceed a certain threshold, $\eps$.
Can we then achieve a higher asymptotic rate in the above sense? Surprisingly, the answer to this question is not known in general. Recent work~\cite{tomamichelww14} at least settles the question in the negative for a class of generalized dephasing channels and in particular for the qubit dephasing channel
\begin{align}
\cZ_{\gamma}: \rho \mapsto (1-\gamma)\rho + \gamma Z \rho Z \,,
\end{align}
where $\gamma \in [0,1]$ is a parameter and $Z$ is the Pauli $Z$ operator. Dephasing channels are particularly interesting examples because dephasing noise is dominant in many physical implementations of qubits.
The results of~\cite{tomamichelww14} thus allow us to fully characterize the achievable region in the limit $n \to \infty$ for such channels, and in particular ensure that
\begin{align}
  \lim_{n \to \infty} \hat{R}_{\mathcal{Z}_{\gamma}}(n;\eps) = I_{c}(\mathcal{Z}_{\gamma})\,,
\end{align}
independent of the value of $\eps \in (0,1)$. Note also that the regularization is not required here since these channels are degradable~\cite{devetakshor05}. 

Here we go beyond studying the problem in the asymptotic limit and develop characterizations of the achievable region for finite values of $n$. We find inner (achievability) and outer (converse) bounds on the boundary of the achievable region. These do not agree for general channels (which is unsurprising given the fact that such an agreement has not even been established asymptotically for nonzero error), but they do coincide for certain important examples.

\renewcommand{\subfigbottomskip}{5pt}
\renewcommand{\subfigcapskip}{12pt}
\renewcommand{\subfiglabelskip}{2pt}

\begin{figure}
\subfigure[Boundary of the achievable region for different values of $n$ (second order approximation).]{
\begin{overpic}[width=0.52\textwidth]{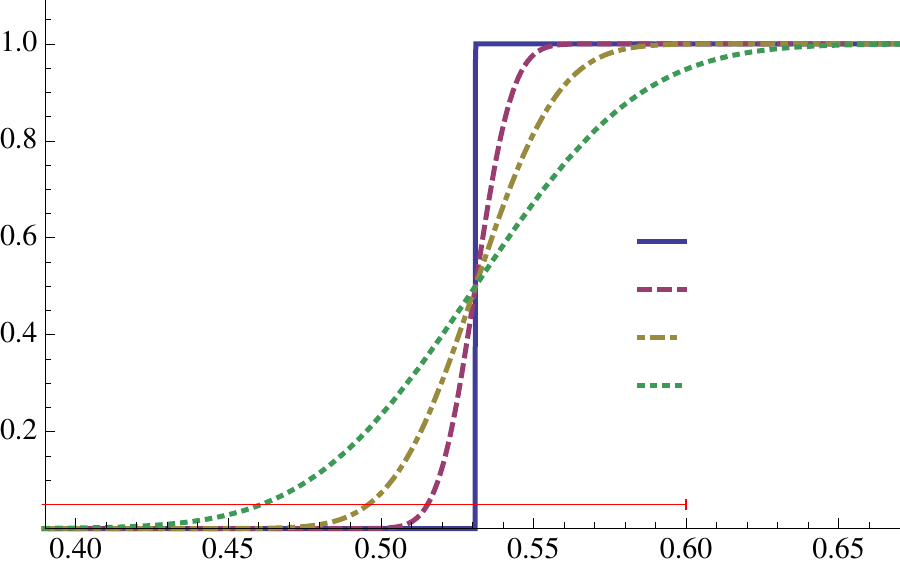}
  \put(79, 35.5){\footnotesize $n \to \infty$}
  \put(79, 30.5){\footnotesize $n = 10^4$}
  \put(79, 25){\footnotesize $n = 2000$}
  \put(79, 19.5){\footnotesize $n = 500$}
  \put(40, -3){\footnotesize rate, $R$}
  \put(-5, 17){\rotatebox{90}{\footnotesize tolerated error, $\eps$}}
  \put(43,23){\vector(-2,1){10}}
  \put(18,31){\footnotesize achievable region}
  \put(35,60){\footnotesize capacity}
  \put(52.9,65){\vector(0,-1){6}}
  \put(8,8.5){\footnotesize \textcolor{red}{cf.\ Subfigure (b)}}
\end{overpic}
\label{fig:dephase1}
}
\subfigure[Boundary of the achievable region for $\eps = 5\%$ (third order approximation in~\eqref{eq:res_z}).]{
\begin{overpic}[width=0.45\textwidth]{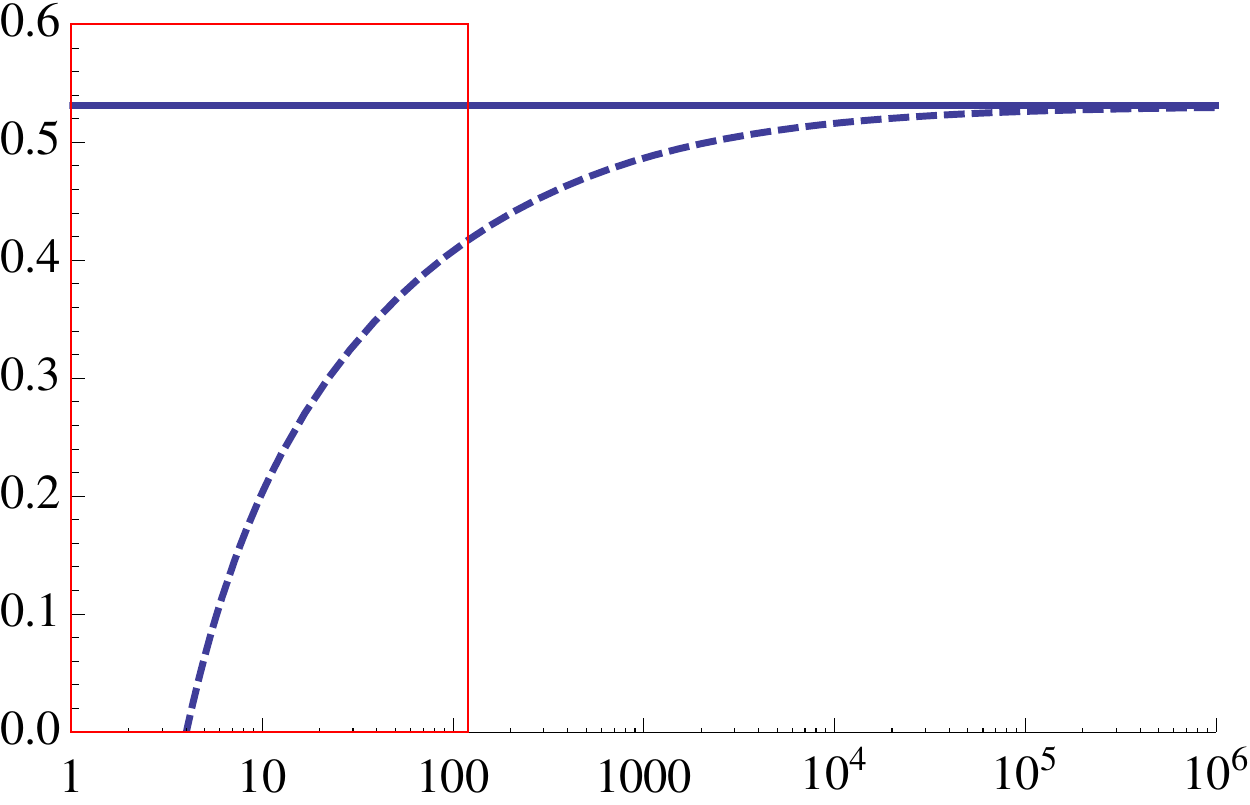}
  \put(30, -4){\footnotesize number of channel uses, $n$}
  \put(-5.5, 27){\rotatebox{90}{\footnotesize rate, $R$}}
  \put(45,44){\vector(2,-1){10}}
  \put(50,44){\footnotesize achievable region}
  \put(32,8){\rotatebox{90}{\footnotesize \textcolor{red}{cf.\ Subfigure (c)}}}
\end{overpic}
\label{fig:dephase2}
}
\subfigure[Comparison of strict bounds with third order approximation for $\eps = 5\%$.]{
\begin{overpic}[width=0.45\textwidth]{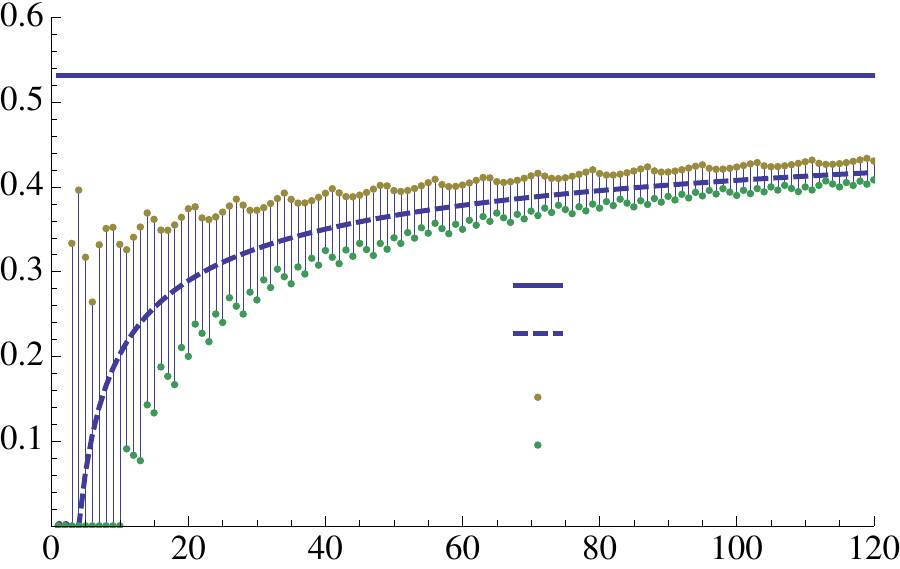}
  \put(30, -4){\footnotesize number of channel uses, $n$} 
  \put(65,31){\footnotesize capacity} 
  \put(65,25.5){\footnotesize $3^{\textrm{rd}}$ order approx.} 
  \put(65,18){\footnotesize exact outer bound} 
  \put(65,12.5){\footnotesize exact inner bound} 
\end{overpic}
\label{fig:dephase3}
}
 \caption{Approximation of the non-asymptotic achievable rate region of a qubit dephasing channel with $\gamma = 0.1$ (see Result~\ref{res:z}).}
 \label{fig:dephase}
\end{figure}
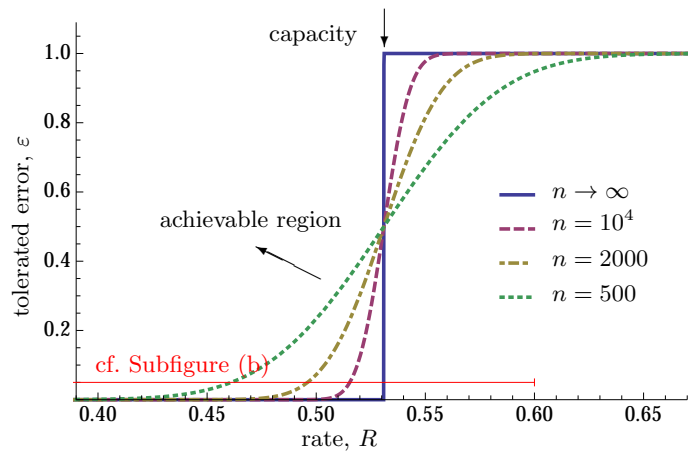
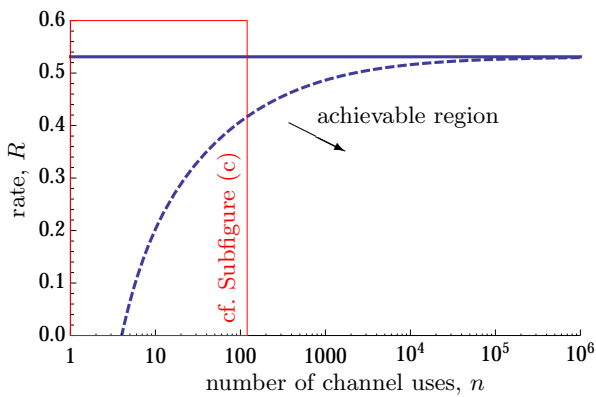
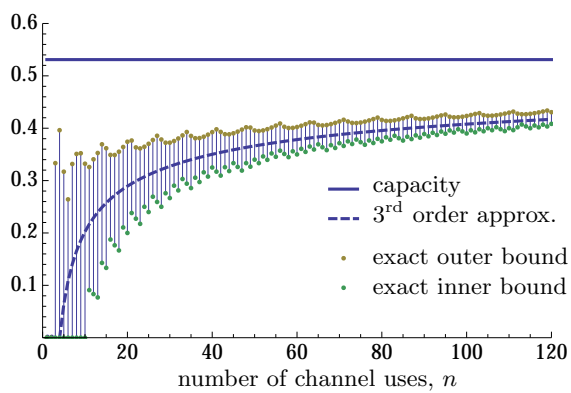


\subsection{Qubit Dephasing Channel}

The first example is the qubit dephasing channel. Building on recent work that established the strong converse for this channel~\cite{tomamichelww14}, we will show that its non-asymptotic achievable region is equivalent to the corresponding region of a (classical) binary symmetric channel. This allows us to employ results from classical information theory and establish the following characterization of the achievable region for the qubit dephasing channel.

\begin{result}\label{res:z}
   For the qubit dephasing channel $\cZ_{\gamma}$ with $\gamma \in [0,1]$, the boundary $\hat{R}(n;\eps)$ satisfies
   \begin{align}\label{eq:res_z}
     \hat{R}(n;\eps) = 1-h(\gamma) + \sqrt{\frac{v(\gamma)}{n}}\, \Phi^{-1}(\eps) + \frac{\log n}{2n}+O\Big(\frac1{n}\Big)\,,
   \end{align}
   where $\Phi$ is the cumulative standard Gaussian distribution, $h(\gamma):= -\gamma \log\gamma-(1-\gamma)\log(1-\gamma)$ denotes the binary entropy and $v(\gamma)$ the corresponding variance, $v(\gamma) := \gamma (\log \gamma + h(\gamma))^2 + (1-\gamma) (\log (1-\gamma) + h(\gamma) )^2$.
\end{result}

The expression without the remainder term $O(\frac1{n})$ is called the third order approximation of the (boundary of the) non-asymptotic achievable region. It is visualized in Figures~\ref{fig:dephase} for an example channel with $\gamma = 0.1$. In Figure~\ref{fig:dephase1} we plot the smallest achievable error $\eps$ as a function of the rate $R$. Here we use the second order expansion without the term $\frac1{2n} \log n$ since it can conveniently be solved for $\eps$. In the limit $n \to \infty$ we see an instantaneous transition of $\eps$ from $0$ to $1$, the signature of a strong converse: coding below the capacity $Q(\cZ_{\gamma}) = 1 - h(\gamma)$ is possible with perfect fidelity whereas coding above the capacity will necessarily result in a vanishing fidelity. 

In Figure~\ref{fig:dephase2} we plot the third order approximation in~\eqref{eq:res_z} for the highest achievable rate, $\hat{R}(n;\eps)$, as a function of $n$ for a fixed fidelity of $95\%$ (i.e.\ we set $\eps = 5\%$). 
For example, this allows us to calculate how many times we need to use the channel in order to approximately achieve the quantum capacity. The third order approximation shows that we need approximately $850$ channel uses to achieve $90\%$ of the quantum capacity. Note that a coding scheme achieving this would probably require us to coherently manipulate $850$ qubits in the decoder, which appears to be a quite challenging task. This example shows that the capacity does not suffice to characterize the ability of a quantum channel to transmit information, and further motivates the study of the achievable region for finite $n$. 

Finally, we remark that the third order approximation is quite strong even for small $n$. To prove this we compare it to concrete upper and lower bounds on $\hat{R}(n;\eps)$ in Figure~\ref{fig:dephase3} and see that the remainder term $O(\frac{1}{n})$ becomes negligible for fairly small $n\approx100$ for the present values of $\gamma$ and $\eps$.


\subsection{Qubit Erasure Channel}

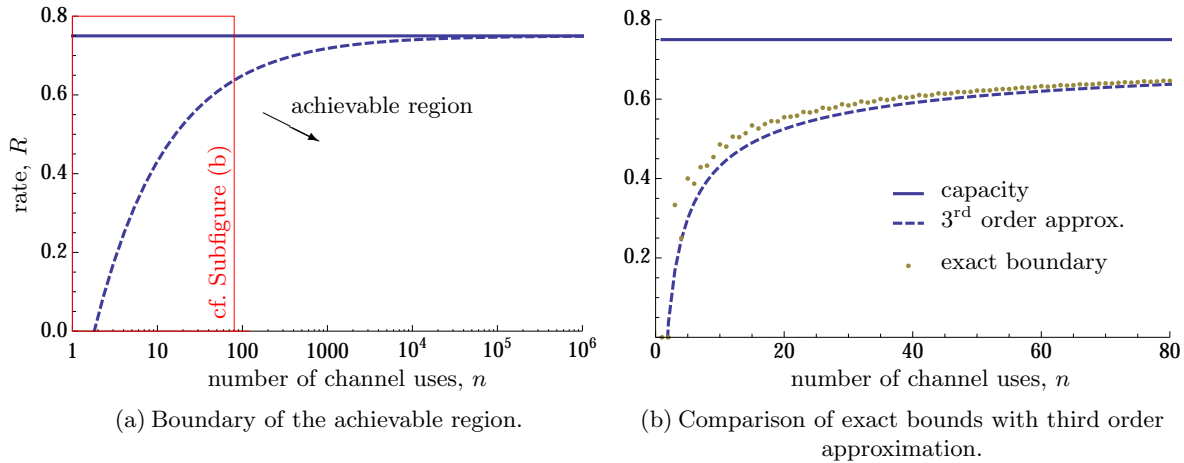
\begin{figure}
\subfigure[Boundary of the achievable region.]{
\begin{overpic}[width=0.45\textwidth]{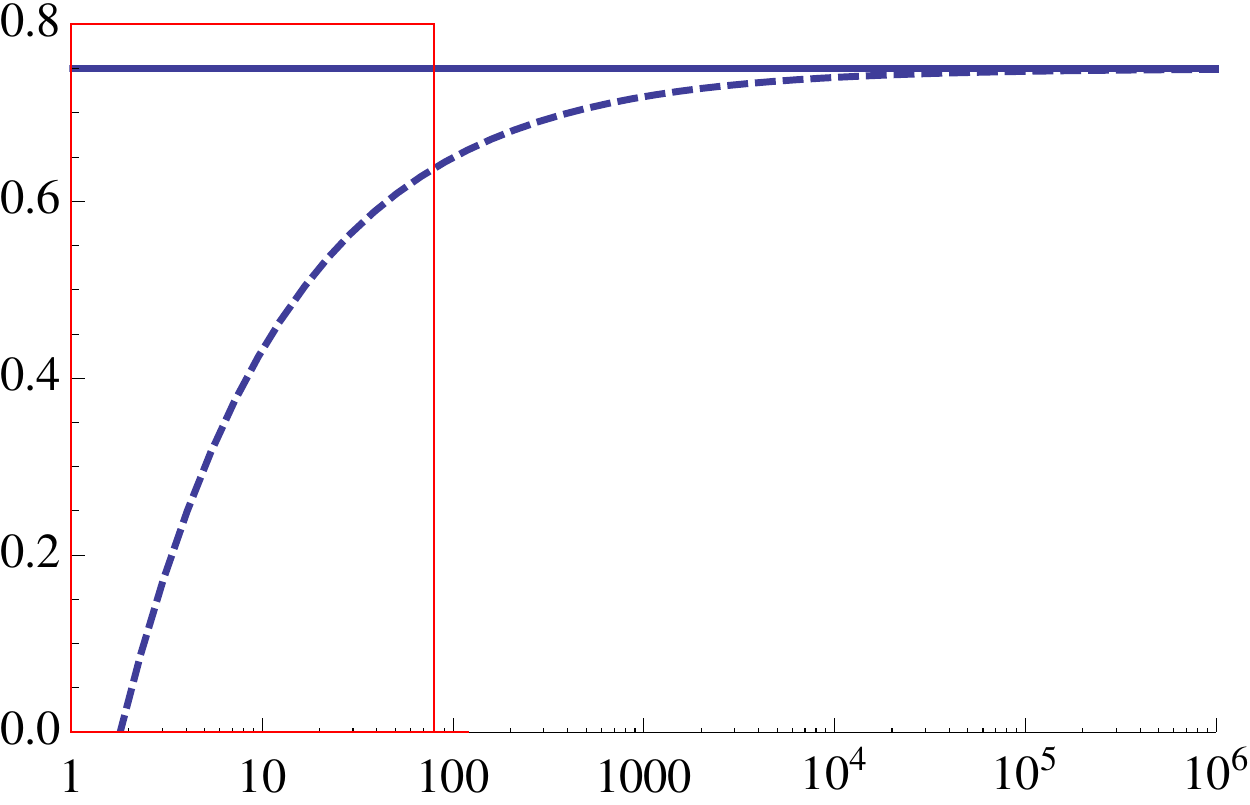} 
  \put(30, -4){\footnotesize number of channel uses, $n$} 
  \put(-5.5, 27){\rotatebox{90}{\footnotesize rate, $R$}}
  \put(40,45){\vector(2,-1){10}}
  \put(45,45){\footnotesize achievable region}
  \put(30,8){\rotatebox{90}{\footnotesize \textcolor{red}{cf.\ Subfigure (b)}}}
\end{overpic}
\label{fig:erasure1}
}
\subfigure[Comparison of exact bounds with third order approximation.]{
\begin{overpic}[width=0.45\textwidth]{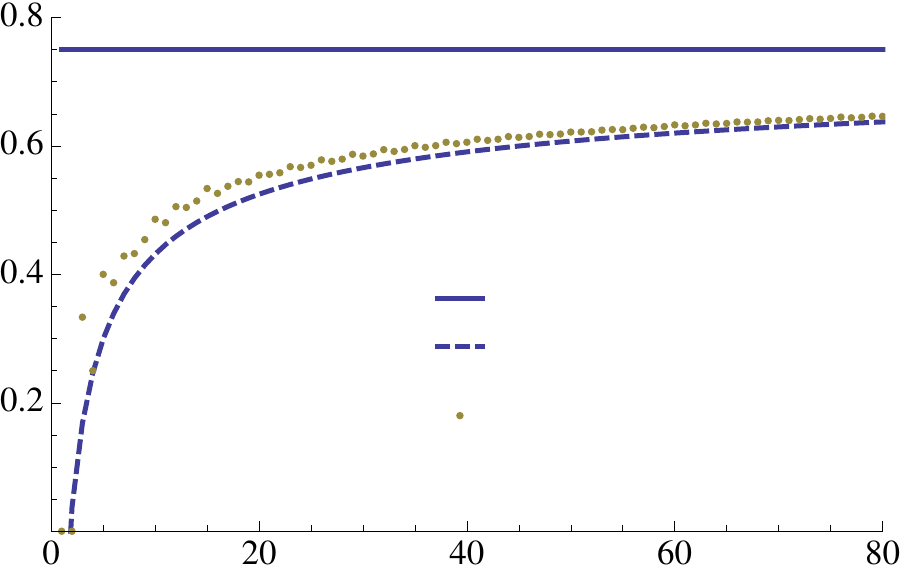} 
  \put(30, -4){\footnotesize number of channel uses, $n$} 
  \put(57,30){\footnotesize capacity} 
  \put(57,24.5){\footnotesize $3^{\textrm{rd}}$ order approx.} 
  \put(57,17){\footnotesize exact boundary} 
\end{overpic}
\label{fig:erasure2}
}
\caption{Approximation of the non-asymptotic achievable rate region (with classical post-processing assistance) of a qubit erasure channel with $\beta = 0.25$ and error parameter $\eps = 1\%$ (see Result~\ref{res:e}).}
\label{fig:erasure}
\end{figure}

Another channel we can analyze in this manner is the qubit erasure channel, given by the map
\begin{align}
\cE_{\beta}:\rho\mapsto(1-\beta)\rho+\beta\proj{e}\,,
\end{align}
where $\beta \in [0,1]$ is a parameter and $\proj{e}$ is a pure state orthogonal to $\rho$. Here we investigate coding schemes that allow classical post-processing (cpp) between the sender and receiver after the quantum transmission (see also Figure~\ref{fig:scheme-cpp} in Section~\ref{sec:notation}). We denote the corresponding boundary of the achievable region by $\hat{R}^{\mathrm{cpp}}(n;\eps)$. Since this includes all codes that do not take advantage of cpp, we clearly have $\hat{R}(n;\eps) \leq \hat{R}^{\mathrm{cpp}}(n;\eps)$ for all channels.

Here we can determine the boundary $\hat{R}^{\mathrm{cpp}}(n;\eps)$ exactly, again by generalizing~\cite{tomamichelww14} and relating the problem to that of the classical erasure channel.

\begin{result}\label{res:e}
   For the qubit erasure channel $\cE_{\beta}$ with $\beta \in [0,1]$, the boundary $\hat{R}^{\mathrm{cpp}}(n;\eps)$ satisfies
   \begin{align}
   \eps = \sum_{l=n-k+1}^n {n \choose l} \beta^l (1-\beta)^{n-l}
\left(1- 2^{n\left(1-\hat{R}^{\mathrm{cpp}}(n;\eps)\right)-l}\right)\,. \label{eq:bec-exact}
   \end{align}
   Moreover, for large $n$, we have the expansion
   \begin{align}
     \hat{R}^{\mathrm{cpp}}(n;\eps) = 1-\beta + \sqrt{\frac{\beta(1-\beta)}{n}} \Phi^{-1}(\eps) + O\Big(\frac{1}{n}\Big)\,.
   \end{align}
\end{result}

The latter expression is a third order approximation of the achievable region, but this time the term proportional to $\frac{1}{n}\log n$ vanishes. In Figure~\ref{fig:erasure} we show this approximation for a qubit erasure channel with $\beta = 0.25$ and fidelity $99\%$. In Figure~\ref{fig:erasure1} we see that the non-asymptotic achievable region reaches $90\%$ of the channel capacity for $n \approx 180$. Again, this confirms that the non-asymptotic treatment is crucial in the quantum setting. In Figure~\ref{fig:erasure2} we compare the third order approximation with the exact boundary of the achievable region in~\eqref{eq:bec-exact}. We see that the approximation is already very precise (and the term $O(\frac{1}{n})$ thus negligible) for fairly small $n \approx 50$.


\subsection{Depolarizing Channel}

\begin{figure}
\subfigure[Inner and outer bounds for $\alpha = 0.05$ and $\eps = 1\%$.]{
    \begin{overpic}[width=0.45\textwidth]{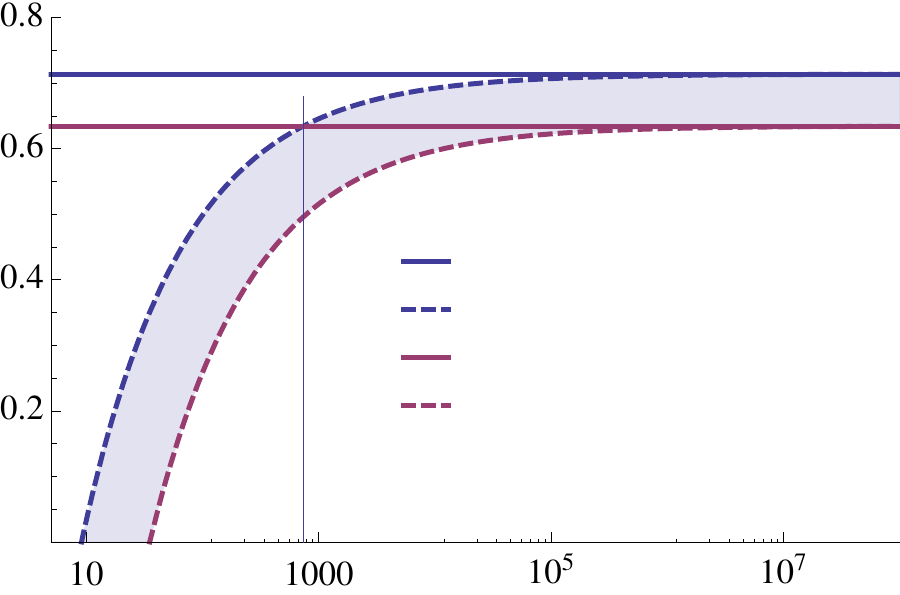}
  \put(30, -4){\footnotesize number of channel uses, $n$} 
  \put(-5.5, 27){\rotatebox{90}{\footnotesize rate, $R$}}
      \put(36,8.5){\footnotesize $N_0 = 738$}
      \put(53,35.5){\footnotesize $Q(\cZ_{\alpha})$}
      \put(53,30.5){\footnotesize outer bound ($2^{\textrm{nd}}$ order)}
      \put(53,25){\footnotesize $I_{\mathrm{c}}(\cD_{\alpha})$}
      \put(53,20){\footnotesize inner bound ($2^{\textrm{nd}}$ order)}
    \end{overpic}
    \label{fig:depol1}
}
\subfigure[Exact outer bound for $\alpha = 0.0825$ and $\eps = 5.5\%$.]{
    \begin{overpic}[width=0.45\textwidth]{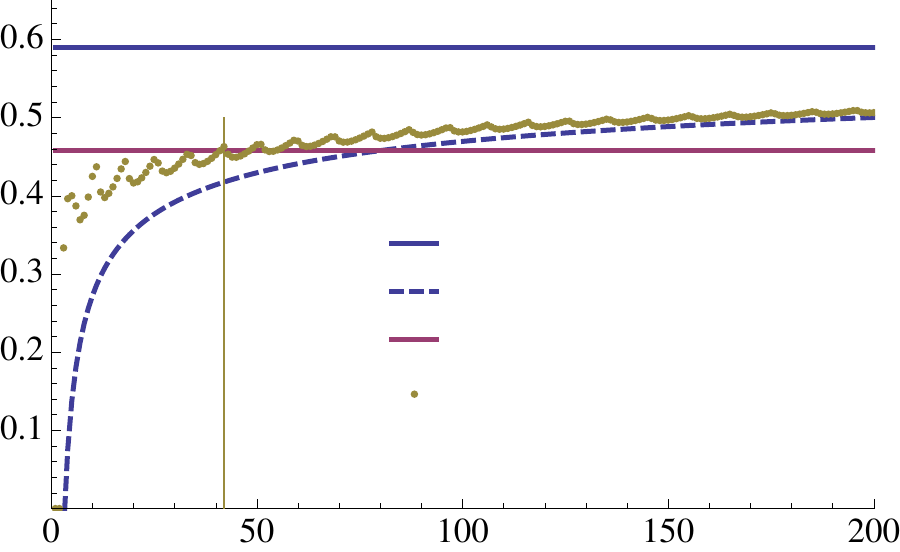}
  \put(30, -4){\footnotesize number of channel uses, $n$} 
      \put(27,8){\footnotesize $N_0 = 42$}
      \put(52,33.5){\footnotesize $Q(\cZ_{\alpha})$}
      \put(52,28){\footnotesize outer bound ($3^{\textrm{rd}}$ order)}
      \put(52,22.5){\footnotesize $I_{\mathrm{c}}(\cD_{\alpha})$}
       \put(52,16.5){\footnotesize exact outer bound}
   \end{overpic}
   \label{fig:depol2}
}
\caption{Approximate inner and outer bounds on the non-asymptotic achievable rate region for the depolarizing channel (see Results~\ref{res:depol} and~\ref{res:inner}).}
\label{fig:depol}
\end{figure}
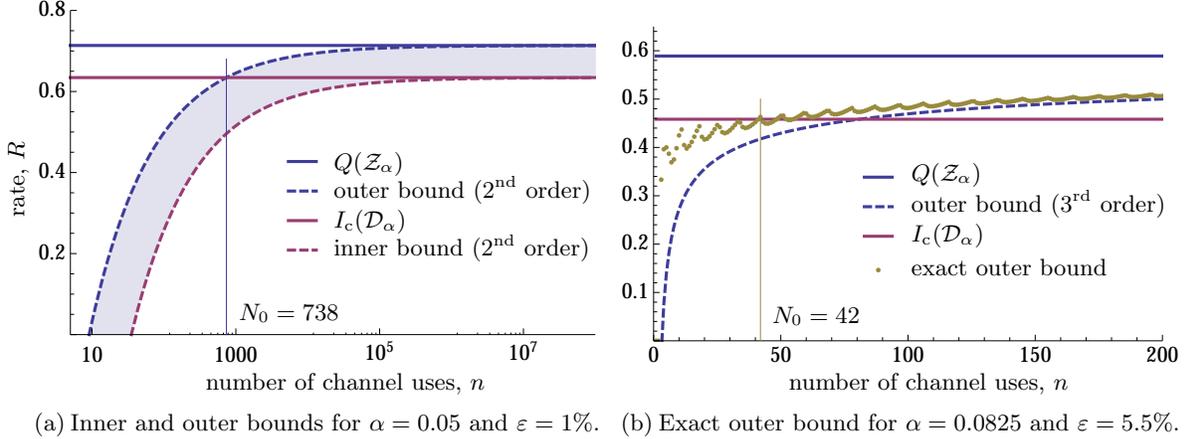

Another basic channel of interest is the qubit depolarizing channel. It is given by the map
\begin{align}
\mathcal{D}_{\alpha}: \rho \mapsto (1-\alpha)\rho + \frac{\alpha}{3}\left(X \rho X + Y \rho Y + Z \rho Z \right)\,,
\end{align}
where $\alpha\in[0,1]$ is a parameter and $X,Y,Z$ are the Pauli operators. For this channel not even a closed formula for the quantum capacity $Q(\mathcal{D}_{\alpha})$ is known, and the non-regularized coherent information
\begin{align}
I_{c}(\mathcal{D}_{\alpha})=1-h(\alpha)-\alpha\log3
\end{align}
is only a strict lower bound on it~\cite{divincenzo98} (where $h(\alpha)$ again denotes the binary entropy). However, various upper bounds on the quantum capacity have been established as well~\cite{smolin08,winter08,scholz14,Rains2001}. For example, in~\cite{smolin08} it is shown that $Q(\mathcal{D}_{\alpha})\leq Q(\mathcal{Z}_{\alpha})=1-h(\alpha)$, the quantum capacity of the qubit dephasing channel with dephasing parameter $\alpha$. Here we extend this result to the non-asymptotic setting and find the following outer (converse) bound for the achievable rate region.

\begin{result}\label{res:depol}
For the qubit depolarizing channel $\mathcal{D}_{\alpha}$ with $\alpha \in [0,1]$, the boundary $\hat{R}^{\mathrm{cpp}}(n;\eps)$ satisfies
\begin{align}
\hat{R}_{\mathcal{D}_{\alpha}}(n;\eps) \leq \hat{R}_{\mathcal{Z}_{\alpha}}(n;\eps)\,,
\end{align}
where $\hat{R}_{\mathcal{Z}_{\alpha}}(n;\eps)$ denotes the boundary of the achievable rate region for the qubit dephasing channel with dephasing parameter $\alpha$ as in Result~\ref{res:z}. 
\end{result}

Clearly this allows us to recycle the bounds in Result~\ref{res:z} and use them as outer bounds for the achievable region. This is done in Figure~\ref{fig:depol} for two example channels. In Figure~\ref{fig:depol1} we plot the second order approximation of the outer bound for a depolarizing channel with $\alpha = 0.05$ and $99\%$ fidelity. We can see that in order to exceed the coherent information, we will need to code for at least $N_0 = 738$ channel uses. This indicates that the question of whether the coherent information is a good or bad lower bound on the asymptotic quantum capacity is  not practically relevant as long as we do not have a quantum computer that is able to perform a decoding operation on many hundreds of qubits. We also show a second order approximation for a general inner bound which is given in Result~\ref{res:inner} below.

In Figure~\ref{fig:depol2} we examine a channel with parameters $\alpha = 0.0825$ and $\eps = 5.5\%$. Instead of using an approximation for the outer bound we use the exact outer bound to give the answer (it is $42$) to the question of how many channel uses we need at minimum to exceed the coherent information. However, note that this does not give us any indication of what code (in particular if it is assisted or not) one would need to use for this purpose. 


\subsection{General Bounds}

We have so far focused our attention on two specific (albeit very important) examples of channels. However, many of the results derived in this paper also hold more generally. For example, we find the following outer (converse) bound for coding schemes that allow classical post-processing.

\begin{result}\label{res:outer}
For any quantum channel $\cN$, the boundary $\hat{R}^{\mathrm{cpp}}(n;\eps)$ satisfies
  \begin{align}
    \hat{R}^{\mathrm{cpp}}(n;\eps) &\leq-\log f\big(\cN^{\otimes n},\eps\big)\,,
  \end{align}
where $f(\cN,\eps)$ is the solution to a semidefinite optimization program. Moreover, if $\cN$ is covariant we find the asymptotic expansion
\begin{align}\label{eq:outer_cov}
     \hat{R}^{\mathrm{cpp}}(n;\eps)\leq\hat{R}_{\mathrm{outer}}^{\mathrm{cpp}}(n;\eps) , \quad \textrm{with} \quad  \hat{R}_{\mathrm{outer}}^{\mathrm{cpp}}(n;\eps)= I_{R}(\cN) + \sqrt{\frac{ V_R^{\eps}( \cN)}{n}} \, \Phi^{-1}(\eps) + O\bigg(\frac{\log n}{n}\bigg)\,,
\end{align}
where the Rains information, $I_R(\cN)$, and its variance, $V_R^{\eps}(\cN)$, are defined in Theorem~\ref{thm:covariant_outer}.
\end{result}

The semidefinite optimization program $f(\cN, \eps)$ (see Section~\ref{sec:converse} for details) is similar in spirit to the metaconverse for classical channel coding by Polyanskiy {\it et al.}~\cite{polyanskiy10}, formulated as a linear program by Matthews~\cite{matthews11}.\footnote{For quantum coding, Matthews and Leung~\cite{matthews14} also give semidefinite optimization program lower bounds on the error boundary $\hat{\eps}_{\cN}(n; R) := \max \{ \eps : (R, n, \eps) \textrm{ is achievable} \}$ for fixed rate $R$.} Note that the bound~\eqref{eq:outer_cov} is tight up to the second order asymptotically for the qubit dephasing channel (Result~\ref{res:z}) and the erasure channel with classical post-processing assistance (Result~\ref{res:e}). However, in the generic covariant case the bound is in general not expected to be tight even in first order. Moreover, if the channel is not covariant we cannot find any non-trivial outer bounds on the non-asymptotic achievable region that allows for an asymptotic expansion in the above sense.


Moreover, an inner (achievability) bound of the form shown in Result~\ref{res:z} also holds generally for all quantum channels.

\begin{result}\label{res:inner}
  For any quantum channel $\cN$, the boundary $R^*(n;\eps)$ satisfies
  \begin{align}\label{eq:general_inner}
    R^*(n;\eps) \geq R^*_{\mathrm{inner}}(n;\eps) , \quad \textrm{with} \quad 
    R^*_{\mathrm{inner}}(n;\eps) = I_{c}(\cN) + \sqrt{\frac{V_{c}^{\eps}(\cN)}{n}} \Phi^{-1}(\eps) + O\left( \frac{\log n}{n} \right)\,,
  \end{align}
where the coherent information, $I_{\textrm{c}}(\cN)$, and its variance, $V_{c}^{\eps}(\cN)$, are defined in Theorem~\ref{th:achieve-second}.
\end{result}

Note that the bound~\eqref{eq:general_inner} is tight up to the second order asymptotically for the qubit dephasing channel (Result~\ref{res:z}). However, the bound does not tightly characterize the achievable region of general channels, although we have reasons to conjecture that it does for degradable channels. In fact, this bound is a direct consequence of an inner bound due to Morgan and Winter~\cite{morgan13} together with a second order expansion of smooth entropies in~\cite{tomamichel12}.



\section{Notation, Information Measures, and Codes}\label{sec:notation}

In this paper $\log$ denotes the binary logarithm. To express the second order expansion of the non-asymptotic quantities we need the \emph{cumulative standard Gaussian distribution} function
\begin{align}
\Phi(x):=\frac{1}{\sqrt{2\pi}}\int_{-\infty}^{x} e^{-\frac{y^{2}}{2}}\mathrm{d}y\,.
\end{align}

We denote \emph{finite-dimensional Hilbert spaces} by capital letters. In particular, we use $A$ and $B$ to model the channel input and output space, whereas $M$ and the isomorphic spaces $M'$ and $M''$ are used to model the quantum systems containing the maximally entangled state to be transmitted. We also use $A^n$ to denote the $n$-fold tensor product of $A$ for any $n \in \mathbb{N}$. The dimension of $A$ is denoted by $|A|$, we use $[A]$ to denote the set $\{1, 2, \ldots, |A|\}$, and fix a standard orthonormal basis $\{ \ket{x}_A \}_{x \in [A]}$. 

We use $\cL(A)$ to denote the set of \emph{linear operators} on $A$, $\cP(A)$ to denote the set of \emph{positive semi-definite operators} on $A$, and $\cS(A) := \{ \rho \in \cP(A) : \tr(\rho) = 1 \}$ to denote \emph{quantum states} with unit trace on $A$. A quantum state is called \emph{pure} if it has rank one.  We write $\rho \ll \sigma$ if the support of $\rho$ is contained in the support of $\sigma$. For general positive operators $\rho,\sigma\in\mathcal{P}(A)$, we define Uhlmann's \emph{fidelity}~\cite{uhlmann85} as
\begin{align}
F(\rho,\sigma):= \left( \left\|\sqrt{\rho}\sqrt{\sigma}\right\|_{1} \right)^2\,,
\end{align}
where $\|X\|_{1}:=\tr\big(\sqrt{XX^{\dagger}}\big)$ is the trace norm. If one of the states is pure, this expression simplifies to $F( |\psi\rangle\!\langle \psi|, \sigma ) = \langle \psi | \sigma | \psi \rangle$.

We often use subscripts to clarify which Hilbert spaces an operator acts on. Let $A'$ be isomorphic to $A$. Throughout the manuscript we denote the \emph{maximally entangled state} on $AA'$ by $\phi_{AA'}=\proj{\phi}_{AA'}$ with $\ket{\phi}_{AA'}=|A|^{-1/2}\sum_{x \in [A]}\ket{x}_{A}\otimes\ket{x}_{A'}$. For a general state $\rho_A\in \cS(A)$, its \emph{canonical purification} is $\ket{\psi^\rho}_{AA'}=|A|\sqrt{\rho_A}\otimes 1_{A'}\ket{\phi}_{AA'}$. We clearly have $\tr_{A'}(\psi^\rho_{AA'}) = \rho_A$ where $\tr_{A'}$ denotes the partial trace over $A'$.

\emph{Quantum channels} are completely positive and trace preserving maps between operators and denoted by calligraphic letters. In particular, we investigate channels $\cN_{A\to B}$ that map $\cP(A)$ to $\cP(B)$. The \emph{Choi state} of a quantum channel $\cN_{A\to B}$ is defined using the corresponding non-calligraphic letter as $N_{AB}=(\mathcal{I}_{A}\otimes\cN_{A'\to B})(|A|\phi_{AA'})$. 
\bigskip


\subsection{Information Measures}

Our asymptotic results are stated in terms of the following quantities.
For $\rho\in\cS(\cH)$ and $\sigma\in\cP(\cH)$ with $\rho \ll \sigma$, Umegaki's \emph{relative entropy}~\cite{umegaki62} and the quantum \emph{relative entropy variance}~\cite{tomamichel12,li12} are given by
\begin{align}
D(\rho\|\sigma):=\tr\left[\rho\left(\log\rho-\log\sigma\right)\right]\quad\mathrm{and}\quad V(\rho\|\sigma)&:=\tr\left[\rho\left(\log\rho-\log\sigma-D(\rho\|\sigma)\right)^{2}\right]\,,
\end{align}
respectively. The \emph{conditional entropy} and the \emph{conditional entropy variance}~\cite{tomamichel12} of a state $\rho_{AB} \in \cS(AB)$ are given as
\begin{align}
H(A|B)_{\rho}:=-D(\rho_{AB}\|1_{A}\otimes\rho_{B})\quad\mathrm{and}\quad V(A|B)_{\rho}:=V(\rho_{AB}\|1_{A}\otimes\rho_{B})\,,
\end{align}
respectively. Related to this we define the \emph{coherent information} of $\rho_{AB}$ as $I(A\rangle B)_{\rho} := -H(A|B)_{\rho}$ and its corresponding variance $V(A\rangle B)_{\rho} := V(A|B)_{\rho}$.

For our non-asymptotic results, we require the following quantities. For $\rho\in\cS(\cH)$ and $\sigma\in\cP(\cH)$ the \emph{hypothesis testing relative entropy}~\cite{wang10} is defined as
\begin{align}\label{eq:betadef}
D_H^{\eps}(\rho\|\sigma) := - \log \beta_{1-\eps}(\rho\|\sigma)\quad\mathrm{with}\quad\beta_{1-\eps}(\rho\|\sigma):= \min_{0 \leq \Lambda \leq 1 \atop \tr[\Lambda\rho] \geq 1-\eps} \tr [\Lambda\sigma]\,.
\end{align}

The \emph{hypothesis testing Rains relative entropy} of a quantum channel $\cN_{A\to B}$ is defined as (following the generalized divergence framework discussed in~\cite{tomamichelww14}),
\begin{align}
\label{eq:rainsinfo}
I_R^{\eps}(\cN_{A\to B}):= \sup_{\rho_{A} \in \cS(A)} I_R^{\eps}(A:B)_{\cN_{A'\to B}(\psi^\rho_{AA'})}\quad\mathrm{with}\quad I_R^{\eps}(A:B)_{\rho}:= \inf_{\sigma_{AB} \in \textrm{PPT}'(A:B)} D_H^{\eps}(\rho_{AB}\|\sigma_{AB})\,,
\end{align}
where PPT$'(A\!:\!B)$ is the \emph{Rains set}~\cite{Rains2001,audenaert02}, a superset of the set of positive partial transpose (PPT) states. It is defined as
\begin{align}
\textrm{PPT}'(A\!:\!B):=\Big\{\tau_{AB}\in\mathcal{P}(AB)\, \Big|\, \big\|T_{B}(\tau_{AB})\big\|_{1}\leq1 \Big\}\,,
\end{align}
where $T_{B}$ denotes the partial transpose map on $B$. In particular, we have the following inequality~\cite[Lm.~2]{rains99}. For every $\sigma_{AB} \in \textrm{PPT}'(A:B)$, we have
\begin{align}\label{eq:rains-ineq}
\langle \phi | \sigma_{AB} | \phi \rangle_{AB} \leq \frac{1}{|M|} 
\end{align}
for all maximally entangled states $\ket{\phi}_{AB}$ of local dimension $|M|$. Finally, a quantum channel $\cN_{A\to B}$ is called PPT preserving if a PPT state input necessarily results in a PPT state output. It turns out that PPT-preserving channels output PPT states for any input, since they have PPT Choi states~\cite{Rains2001} (see the discussion after Eq.\ 4.13). Channels with PPT Choi states were also called PPT-binding in~\cite{horodecki_binding_2000}.
\bigskip



\subsection{Codes for Entanglement Transmission Assisted by Classical Post-Processing}
\label{sec:notation-codes}

We have defined unassisted entanglement-transmission codes in Section~\ref{sec:results} and in Figure~\ref{fig:scheme}. Let us reintroduce them in the context of codes assisted by classical post-processing.

\begin{figure}
\begin{overpic}[width=.5\textwidth]{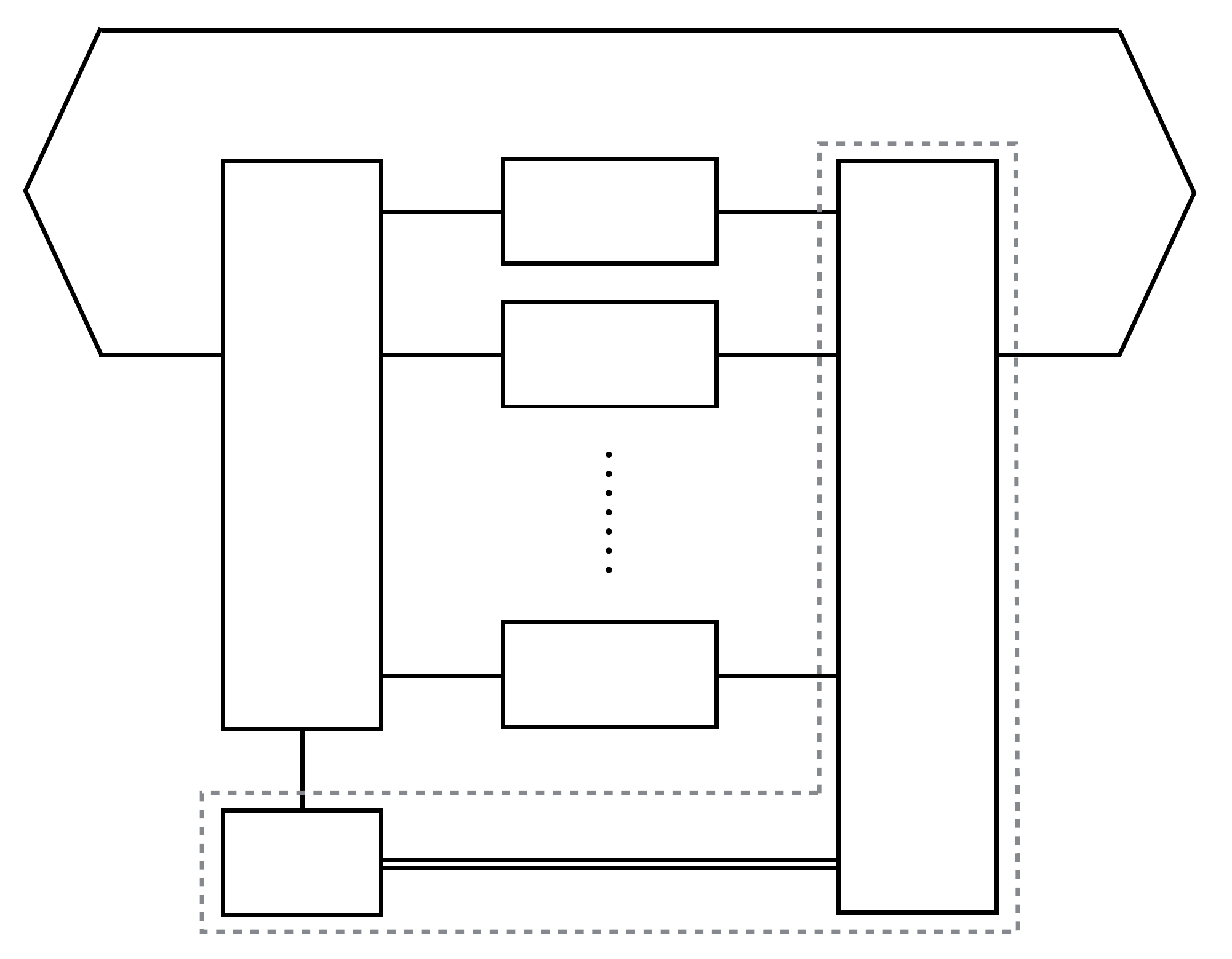}
  \put(24,40){\footnotesize $\cE$}
  \put(73,34){\footnotesize$\cD$}
  \put(23,6.5){\footnotesize$\cD$}
  \put(48,21.5){\footnotesize$\cN$}
  \put(48,47){\footnotesize$\cN$}
  \put(48,58.5){\footnotesize$\cN$}
  \put(-9,61){\footnotesize${\phi_{MM'}}$}
  \put(98,61){\footnotesize${\rho_{MM''}}$}
  \put(11,50.5){\footnotesize$M'$}
  \put(11,71){\footnotesize$M$}
  \put(84,50.5){\footnotesize$M''$}
  \put(20,14.8){\footnotesize$Q$}
  \put(34.5,24.5){\footnotesize$A_n$}
  \put(34.5,50.5){\footnotesize$A_2$}
  \put(34.5,62){\footnotesize$A_1$}
  \put(60,24.5){\footnotesize$B_n$}
  \put(60,50.5){\footnotesize$B_2$}
  \put(60,62){\footnotesize$B_1$}
\end{overpic}
  \caption{Coding Scheme for entanglement transmission over $n$ uses of a channel $\cN_{A\to B}$ with classical post-processing. The encoder $\cE \equiv \cE_{M' \to A^n Q}$ encodes $M'$ into the channel input systems and a local memory $Q$. Later, the decoder $\cD \equiv \cD_{Q B^n \to M''}$ recovers the maximally entangled state from the channel output systems and the memory $Q$ using classical communication and local operations. The performance of the code is measured using the fidelity $F(\phi_{MM'}, \rho_{MM''})$.}
  \label{fig:scheme-cpp}
\end{figure}
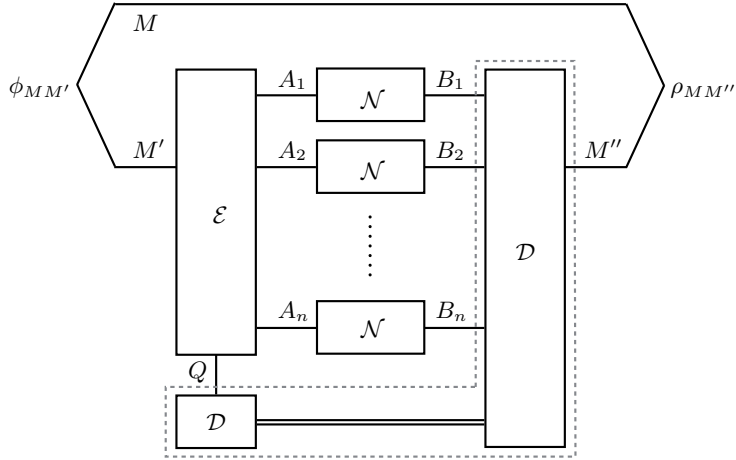

For this, we again consider any quantum channel $\cN \equiv \cN_{A \to B}$ and its $n$-fold extension $\cN^{\otimes n}$ that maps states on $A^n$ to states on $B^n$. An \emph{entanglement transmission code assisted by classical post-processing} for $\cN^{\otimes n}$ is given by a triplet $\{ |M|, \cE, \cD\}$, as depicted in Figure~\ref{fig:scheme-cpp}. Here, $|M|$ is the local dimension of a maximally entangled state $\ket{\phi}_{MM'}$ that is to be transmitted over $\cN^{\otimes n}$. The encoder $\cE_{M' \to A^n Q}$ is a completely positive trace-preserving map that prepares the channel inputs $A_1, A_2, \ldots A_n$ and a local memory system, which we denote by $Q$. The decoder $\cD_{Q B^n \to M''}$ is a completely positive trace-preserving map that is restricted to local operations and classical communication with regards to the bipartition $Q:B^n$ and outputs $M''$ on the receiver's side.

Finally, we note that unassisted codes are recovered if we choose $Q$ to be trivial. Hence, unassisted codes are contained in the set of assisted codes.


\section{Proofs: General Bounds}

It is convenient to first discuss the general bounds discussed in Results~\ref{res:outer} and~\ref{res:inner}.

\subsection{General Outer Bounds on the Achievable Region}\label{sec:converse}

In this section we derive the general outer bounds from Result~\ref{res:outer}, precisely stated as Corollary~\ref{cor:relaxation} and Theorem~\ref{thm:covariant_outer} below. Our results are inspired by the strong converse results for generalized dephasing channels from~\cite{tomamichelww14} and the metaconverse for classical channel coding~\cite{polyanskiy10}. 

We first formulate a general metaconverse bounding possible rates $R$ given a tolerated error $\eps$ for single uses of a fixed channel $\cN$. This bound has the useful property that channel symmetries can be used to simplify its form. Nevertheless, when applied to $n$ instances of $\cN$, the bound is not efficiently computable. Loosening the bound produces a more computationally tractable convex optimization, specifically a semidefinite program.

\begin{lemma}\label{lem:metaconverse}
Let $\mathcal{N}_{A\to B}$ be a quantum channel. Then, for any fixed $\eps \in(0,1)$, the achievable region with cpp-assistance satisfies,\footnote{Indeed, the outer bound also holds for coding schemes with (unphysical) positive partial transpose (PPT) assistance and this includes in particular classical pre- and post-processing assistance (see, e.g., \cite{matthews14} for a precise definition of PPT assisted codes).\label{ft:ppt}}
\begin{align}\label{eq:meta_first}
\hat{R}^{\mathrm{cpp}}(1;\eps)\leq  I_R^{\eps}\big(\cN\big)\,.
\end{align}
\end{lemma}

\begin{proof}
First, observe that the encoding operation $\cE_{M'\to AQ}$ can be chosen to be an isometry without loss of generality, because we may include any extension systems needed for the Stinespring dilation into $Q$.
Then we may express the entanglement fidelity as follows
\begin{align}
F 
&= \tr[\phi_{MM'}\cD_{BQ\to M'}\circ \cN_{A\to B}\circ \cE_{M'\to AQ}(\phi_{MM'})]\\
&=\tr[\cE_{M\to \bar A\bar Q}\otimes \cD_{BQ\to M'}^\dagger (\phi_{MM'})\cN_{A\to B}(\cE_{M'\to AQ}\otimes \cE_{M\to \bar A\bar Q}(\phi_{MM'}))]\,.
\end{align} 
Since $\cE$ is an isometry, the state $\rho_{A\bar AQ\bar Q}=\cE_{M'\to AQ}\otimes \cE_{M\to \bar A\bar Q}(\phi_{MM'})$ is pure, and therefore there exists an isometry $W_{A'\to \bar A Q\bar Q}$ such that $\ket{\rho}_{A\bar AQ\bar Q}=W_{A'\to \bar A Q\bar Q}\ket{\psi^{\rho}}_{AA'}$. Thus,
\begin{align}
F=\tr[W_{A'\to \bar A Q\bar Q}^\dagger (\cE_{M\to \bar A\bar Q}\otimes \cD_{BQ\to M'}^\dagger (\phi_{MM'}))W_{A'\to \bar A Q\bar Q} \cN_{A\to B}(\psi^{\rho}_{AA'})]\,.
\end{align}
Now consider the entanglement fidelity of any PPT$'$ state $\sigma_{A'B}$ instead of $\cN_{A\to B}(\psi^{\rho}_{AA'})$. By  \eqref{eq:rains-ineq} we have
\begin{align}\label{eq:protometa}
\tr[\phi_{MM'}(\cE_{M\to \bar A\bar Q}^\dagger\otimes \cD_{BQ\to M'}(W_{A'\to \bar A Q\bar Q} \sigma_{A'B} W_{A'\to \bar A Q\bar Q}^\dagger)) ]\leq \frac 1M\,,
\end{align}
as the operations on $\sigma_{A'B}$ are all PPT-preserving.
We may write this bound in terms of the hypothesis-testing relative entropy, because 
\begin{align}
\Lambda_{A'B}:=W_{A'\to \bar A Q\bar Q}^\dagger (\cE_{M\to \bar A\bar Q}\otimes \cD_{BQ\to M'}^\dagger (\phi_{MM'}))W_{A'\to \bar A Q\bar Q}
\end{align}
is a feasible test to discriminate between $\cN_{A\to B}(\psi^{\rho}_{A'A}))$ and $\sigma_{A'B}$. That is, $\Lambda_{A'B}$ satisfies $0\leq \Lambda_{A'B}\leq 1_{A'B}$ and $\tr[\Lambda_{A'B}\cN_{A\to B}(\psi^{\rho}_{A'A}))]\geq 1-\eps$, the former since $\cD$ is completely-positive and trace-preserving and $\cE$ and $W$ are isometries, the latter by assumption that $F\geq 1-\eps$. From \eqref{eq:protometa} we then obtain
\begin{align}
\hat{R}^{\mathrm{cpp}}(1;\eps) \leq D_{H}^{\eps}(\cN_{A\to B}(\psi^{\rho}_{A'A}))\,\|\, \sigma_{A'B})\,.
\end{align}
Since the bound holds for all PPT$'$ $\sigma_{A'B}$, we may take the infimum over this set to obtain
\begin{align}
\hat{R}^{\mathrm{cpp}}(1;\eps) \leq I_R^{\eps}(A': B)_{\cN_{A\to B}(\psi^{\rho}_{A'A})}\,.
\end{align}
This bound depends on the precise channel input $\rho_A\in\cS(A)$ used by the code, but we can remove the dependence by taking the supremum over all possible inputs. The result is \eqref{eq:meta_first}.
\end{proof}

Applied to the channel $\cN^{\otimes n}$ we immediately get for any fixed $\eps \in(0,1)$, 
\begin{align}
\hat{R}^{\mathrm{cpp}}(n;\eps)\leq  I_R^{\eps}\big(\cN^{\otimes n}\big)\,.
\end{align}
This bound forms a counterpart to Lemma~\ref{lm:morgan}, but suffers from the same weakness. It is generally hard to evaluate this bound even for moderately large $n$. However, we may relax the bound from Lemma~\ref{lem:metaconverse} to a convex optimization by restricting the form of the possible states $\sigma_{AB}$ in the definition of the hypothesis testing Rains relative entropy $I_R^{\eps}$.

\begin{corollary}\label{cor:relaxation}
Let $\mathcal{N}_{A\to B}$ be a quantum channel. We define the function
\begin{align}\label{eq:f_function}
f(\cN,\eps)&:=\inf_{\rho_{A}\in \cS(A)}\inf_{\Lambda_{AB}\in \Gamma(\rho_A,\cN,\eps)}\sup_{\cM_{A\to B}\in \text{PPT}} \tr[\Lambda_{AB}M_{AB}]\,,
\end{align}
with the set $\Gamma(\rho_A,\cN,\eps):=\{\Lambda_{AB}:0\leq \Lambda_{AB}\leq \rho_A^T\otimes 1_B,\tr[\Lambda_{AB}N_{AB}]\geq 1-\eps\}$, and the Choi states $M_{AB}$ of $\cM_{A\to B}$ and $N_{AB}$ of $\cN_{A\to B}$. Then, for any fixed $\eps\in(0,1)$, the achievable region satisfies 
\begin{align}
\hat{R}^{\mathrm{cpp}}(1;\eps) &\leq-\log f(\cN,\eps)\,.
\end{align}
\end{corollary}

\begin{proof}
Suppose that $\sigma_{AB}=(\mathcal{I}_{A}\otimes\cM_{A'\to B})(\psi^{\rho}_{AA'})$ for some PPT-preserving (PPT-binding) channel $\cM_{A\to B}$. Writing out the right-hand side of~\eqref{eq:meta_first} using this relaxation and notation gives
\begin{align}
\hat{R}^{\mathrm{cpp}}(1;\eps)\leq -\log\inf_{\rho_{A}\in \cS(A)}\sup_{\cM_{A\to B}\in \text{PPT}}\inf_{0 \leq \Lambda' \leq 1 \atop \tr[\Lambda'\rho]\geq1-\eps}\tr\big[\Lambda'_{AB}(\mathcal{I}_{A}\otimes\mathcal{M}_{A'\to B})(\psi^{\rho}_{AA'})\big]\,.
\end{align}
Now we may define $\Lambda_{AB}=(\rho_A^T)^{1/2}\Lambda'_{AB}(\rho_A^T)^{1/2}$ and find
\begin{align}\label{eq:proof_joe}
\hat{R}^{\mathrm{cpp}}(1;\eps)\leq-\log\inf_{\rho_{A}\in \cS(A)}\sup_{\cM_{A\to B}\in \text{PPT}}\inf_{\Lambda_{AB}\in \Gamma(\rho_A,\cN,\eps)}\tr[\Lambda_{AB}M_{AB}]\,.
\end{align}
Finally for fixed channel input $\rho_{A}$, we can reverse the order of the inner optimizations in~\eqref{eq:proof_joe} by von Neumann's minimax theorem, since the objective function is linear and the sets are both convex and compact. This concludes the proof.
\end{proof}

Furthermore, we show in Appendix~\ref{app:sdp} that $f(\cN,\eps)$ can be expressed as a semidefinite optimization program that satisfies strong duality.


\paragraph*{Group Covariant Channels:}

In the following we show that symmetries of the channel can further simplify the outer bounds. First let us state precisely what we mean by symmetries. Suppose $G$ is a group represented by unitary operators $U_g$ on $A$ and $V_g$ on $B$. A quantum channel $\cN_{A\to B}$ is covariant with respect to $G$ when  
\begin{align}
V_g \cN(\cdot)V_g^\dagger = \cN(U_g \cdot U_g^\dagger),\quad \forall g\in G\,.
\end{align}
Alternatively we can also write this as an invariance of the channel
\begin{align}
\cN(\cdot)= V_g^\dagger\cN(U_g \cdot U_g^\dagger)V_g,\quad \forall g\in G\,.
\end{align}
Now the main workhorse to simplify our outer bounds for channels with symmetries is~\cite[Prop.~2]{tomamichelww14}, which states that we may restrict the optimization in Lemma~\ref{lem:metaconverse} to covariant input states. Due to the form of the hypothesis testing Rains relative entropy, we may then also choose group covariant PPT$'$ states $\sigma$ and test operators $\Lambda$ to obtain the tightest bound. Note that the convex optimization outer bound in Corollary~\ref{cor:relaxation} inherits these symmetry simplifications.

For general tensor product channels, which are invariant to permutation of the inputs and outputs, this allows us to restrict attention to pure states that are permutation invariant. Moreover, if the channel is covariant, that is, covariant with respect to the full unitary group, then the channel input state can be chosen to be maximally mixed. 

Now let $\cN_{A\to B}$ be a covariant quantum channel and $\phi_{AA'}$ a maximally entangled state. Then, we bound
\begin{align}\label{eq:covariant_dephasing}
\hat{R}^{\mathrm{cpp}}(n;\eps) &\leq \min_{\sigma_{AB} \in \mathrm{PPT'}(A:B)} \frac{1}{n} D_H^{\eps} \big( \cN_{A'\to B}(\phi_{AA'})^{\otimes n} \big\| \sigma_{AB}^{\otimes n} \big)\,,
\end{align}  
where we voluntarily restricted the minimization to product states $\sigma_{AB}^{\otimes n}$ in PPT$'(A\!:\!B)$. Moreover, since these states have tensor product structure, the outer bound can be expanded using~\cite{tomamichel12,li12}
\begin{align}\label{eq:expand-hypo}
\frac{1}{n}D_H^{\eps}(\rho^{\otimes n}\|\sigma^{\otimes n})= D(\rho\|\sigma)+\sqrt{\frac{V(\rho\|\sigma)}{n}} \,\Phi^{-1}(\eps) + O\left(\frac{\log n}{n}\right)\,.
\end{align}
This leads to the following theorem.

\begin{theorem}\label{thm:covariant_outer}
  Let $\cN \equiv \cN_{A\to B}$ be a quantum channel. We define its Rains information as
  \begin{align}
    I_{R}(\cN) := \min_{\sigma_{AB} \in \mathrm{PPT}'(A:B)} D( \cN_{A'\to B}(\phi_{AA'}) \| \sigma_{AB}) \,.
  \end{align}
  and let $\Pi \subset \mathrm{PPT}'(A:B)$ be the set of states that achieve the minimum. The variance of the channel Rains information is
  \begin{align}
   V_{R}^{\eps}(\cN) := \begin{cases} \displaystyle \max_{\sigma_{AB} \in \Pi} V (\cN_{A'\to B}(\phi_{AA'}) \| \sigma_{AB}) & \mathrm{for} \quad \eps < \frac{1}{2} \\
   \displaystyle \min_{\sigma_{AB} \in \Pi} V (\cN_{A'\to B}(\phi_{AA'}) \| \sigma_{AB}) & \mathrm{for} \quad \eps \geq \frac{1}{2} \end{cases}\,.
   \end{align}
If $\cN$ is covariant, then for any fixed $\eps \in(0,1)$, the achievable region with cpp-assistance satisfies
\begin{align}
     \hat{R}^{\mathrm{cpp}}(n;\eps) &\leq \hat{R}_{\mathrm{outer}}^{\mathrm{cpp}}(n;\eps) , \quad \textrm{with} \quad  
     \hat{R}_{\mathrm{outer}}^{\mathrm{cpp}}(n;\eps) = I_{R}(\cN) + \sqrt{\frac{ V_R^{\eps}( \cN)}{n}} \, \Phi^{-1}(\eps) + O\bigg(\frac{\log n}{n}\bigg)\,.
\end{align}
\end{theorem}

Since we are here interested in outer bounds, we are also free to chose a potentially sub-optimal $\sigma_{AB} \in \textrm{PPT}'(A:B)$ to further relax this bound. As we will see in Section~\ref{sec:examples} for the qubit dephasing channel and the erasure channel with classical post-processing assistance the bound from Theorem~\ref{thm:covariant_outer} is tight up to the second order asymptotically.


\subsection{General Inner Bounds on the Achievable Region}\label{sec:achieve}

In this section we derive the general inner bound from Result~\ref{res:inner}, stated as Theorem~\ref{th:achieve-second} below. We use the decoupling approach~\cite{dupuis10,dupuis09,hayden08}, and in particular a one-shot bound by Morgan and Winter~\cite{morgan13} which is a tighter version of previous bounds~\cite{berta08,datta09}.

To reproduce their result we need the following additional notation. Sub-normalized quantum states are collected in the set $\cS_{\bullet}(A) := \{ \rho \in \cP(A) : \tr(\rho) \leq 1 \}$. The purified distance~\cite{tomamichel09} $\eps$-ball around $\rho\in\cS(A)$ is then defined as $\mathcal{B}^{\eps}(\rho):=\left\{\bar{\rho}\in\mathcal{S}_{\bullet}(\mathcal{H}) \,\middle|\, F(\bar{\rho},\rho) \geq (1-\eps)^2 \right\}$. Finally, for $\rho_{AB}\in\cS(AB)$ and $\eps\geq0$ the \emph{smooth conditional min-entropy}~\cite{renner05,tomamichel09} is defined as
\begin{align}
H_{\min}^{\eps}(A|B)_{\rho}:= \sup_{\bar{\rho}_{AB} \in \mathcal{B}^{\eps}(\rho_{AB})} \sup_{\sigma_{B} \in \cS(B)} \sup\left\{\lambda\in\mathbb{R} \,\middle|\, \bar{\rho}_{AB}\leq 2^{-\lambda} \cdot1_{A}\otimes\sigma_{B}\right\}\,.
\end{align}

Let us now restate Morgan and Winter's result expressed in terms of the non-asymptotic achievable region as introduced in Section~\ref{sec:results}.

\begin{lemma}\cite[Prop.~20]{morgan13}\label{lm:morgan}
Let $\mathcal{N}_{A\to B}$ be a quantum channel with complementary channel $\cN_{A\to E}^{c}$. Then $(R,1,\eps)$ is achievable if, for any $\eta \in (0,\eps]$ and any state $\rho_A \in \cS_{\circ}(A)$, we have
\begin{align}
R \leq H_{\min}^{\sqrt{\eps}-\eta}(A|E)_{\omega} - 4\log \frac{1}{\eta}\,,
\end{align}
where $\omega_{AE} = \big(\mathcal{I}_{A} \otimes \mathcal{N}_{A'\to E}^{c})(\psi_{AA'}^{\rho})$.
\end{lemma}

Note that Morgan and Winter use the purified distance as their figure of merit whereas we use the fidelity criterion~\eqref{eq:coding_def}. This accounts for the square root in the smoothing parameter of the conditional min-entropy. Also Morgan and Winter state their result for the special case $n = 1$, but this can be generalized to arbitrary $n \in \mathbb{N}$ if we simply consider $\cN_{A\to B}^{\otimes n}$ as a single channel. This leads immediately to the following inner bound on the achievable region:

\begin{corollary}
  Using the notation of Lemma~\ref{lm:morgan} with $\omega_{A^nE^n} = \big(\mathcal{I}_{A^n} \otimes (\mathcal{N}_{A'\to E}^{c})^{\otimes n} \big)(\psi_{A^nA^{'n}}^{\rho})$, we have
  \begin{align}
  \hat{R}(n;\eps) \geq \sup_{\eta \in (0,\eps)} \sup_{\rho_{A^n} \in \cS(A^n)} 
    \frac{1}{n} \big( H_{\min}^{\sqrt{\eps}-\eta}(A^n|E^n)_{\omega}   -4 \log \frac{1}{\eta} -1 \big)\,.
  \end{align}
\end{corollary}

The problem with this bound is that it is generally hard to evaluate, even for moderately large values of $n$. Hence we are interested to further simplify the expression on the right-hand side in this regime. To do so, we choose $\eta = 1/{\sqrt{n}}$ and use input states of the form $\rho_A^{\otimes n}$. This yields the following relaxation, which holds if $n > \frac{1}{\eps}$:
\begin{align}\label{eq:achieve1}
  \hat{R}(n;\eps) \geq  
  \sup_{\rho_A \in \cS(A)} 
    \frac{1}{n} \big( H_{\min}^{\eps_n}(A^n|E^n)_{\omega^{\otimes n}}  -2 \log n -1 \big) \,.
\end{align}
Here we introduced $\eps_n = \sqrt{\eps}-\frac{1}{\sqrt{n}}$ and $\omega_{AE}$ as in Lemma~\ref{lm:morgan}. Using standard second order expansion methods~\cite{tomamichel12}, we can give an asymptotic expansion of $R_{\mathrm{inner}}^*(n;\eps)$ in~\eqref{eq:achieve1} as follows.

\begin{theorem}\label{th:achieve-second}
  Let $\cN \equiv \mathcal{N}_{A\to B}$ be a quantum channel. We define its coherent information as
  \begin{align}
I_{c}(\cN) := \max_{\rho_A \in \cS(A)} I(A \rangle B)_{\omega}\,, \quad \textrm{with} \quad  
\omega_{AB} = \big(\mathcal{I}_{A} \otimes \mathcal{N}_{A'\to B})(\psi_{AA'}^{\rho})
  \end{align}
  and let $\Pi \subset \cS_{\circ}(A)$ be the set of states that achieve the maximum. Define
  \begin{align}
    V^{\eps}_{c}(\cN) :=
\begin{cases}
 \displaystyle \min_{\rho_A \in \Pi} V(A \rangle B)_{\omega} & \mathrm{for} \quad \eps < \frac{1}{2} \\
\displaystyle \max_{\rho_A \in \Pi} V(A \rangle B)_{\omega} & \mathrm{for} \quad \eps \geq \frac{1}{2} 
\end{cases}\,.
  \end{align}
  Then, for any fixed $\eps\in(0,1)$, the achievable region satisfies 
  \begin{align}
    \hat{R}(n;\eps) \geq \hat{R}_{\mathrm{inner}}(n;\eps) , \quad \textrm{with} \quad  
    \hat{R}_{\mathrm{inner}}(n;\eps) = I_{c}(\cN) + \sqrt{\frac{V_{c}^{\eps}(\cN)}{n}}\,\Phi^{-1}(\eps) + O\left(\frac{\log n}{n} \right)\,.
  \end{align}
\end{theorem}

\begin{proof}
We analyze the expression in~\eqref{eq:achieve1} using the following asymptotic expansion of the smooth conditional min-entropy~\cite{tomamichel12},
\begin{align}\label{eq:expand-min}
\frac{1}{n}H_{\min}^{\eps}(A|B)_{\rho^{\otimes n}}= H(A|B)_{\rho}+\sqrt{\frac{V(A|B)_{\rho}}{n}} \,\Phi^{-1}\big(\eps^2 \big) +O\left(\frac{\log n}{n}\right)\,.
\end{align}
This yields that for any $\rho_A \in \cS(A)$, we have
\begin{align}
\hat{R}(n;\eps) \geq  H(A|E)_{\omega} + \sqrt{\frac{V(A|E)_{\omega}}{n}}\, \Phi^{-1}(\eps) + O\left(\frac{\log n}{n}\right)\,,
\end{align}
and then by duality of the conditional entropy we find $H(A|E)_{\omega} = I(A\rangle B)_{\omega}$. Furthermore, it is easy to verify that $V(A\rangle E)_{\omega} = V(A\rangle B)_{\omega}$ (see, e.g.,~\cite{hayashitomamichel14}). We conclude the proof by choosing an optimal state $\rho_A \in \Pi$ depending on the sign of $\Phi^{-1}(\eps)$.
\end{proof}

As we will see in Section~\ref{sec:examples} for the qubit dephasing channel the bound in Theorem~\ref{th:achieve-second} agrees with the outer bound stated in Result~\ref{res:outer} up to the second order asymptotically.


\section{Proofs: Examples}\label{sec:examples}

In this section we derive our results concerning the qubit dephasing channel, the erasure channel, and the qubit depolarizing channel as announced in Results~\ref{res:z},~\ref{res:e}, and~\ref{res:depol}.


\subsection{Covariant Generalized Dephasing Channels} 

First we consider covariant generalized dephasing channels, which have the (additional) property that $\cZ_{A'\to B}(\psi_{AA'})$ has full support on the projector $\Pi = \sum_x |x\rangle\!\langle x|_A \otimes |x\rangle\!\langle x|_B$ in some basis. In that case, starting from~\eqref{eq:covariant_dephasing}, we can use the data-processing inequality for a map $\cE(\cdot) = \Pi \cdot \Pi + (1-\Pi) \cdot (1-\Pi)$ to write
\begin{align}
\hat{R}^{\mathrm{cpp}}(n;\eps)
&\leq\min_{\sigma_{AB} \in \textrm{PPT}'(A:B)} \frac{1}{n} D_H^{\eps} \big( (\cZ_{A'\to B}(\phi_{AA'}))^{\otimes n} \big\| \sigma_{AB}^{\otimes n}\big)\\
&\leq \min_{\sigma_{B} \in \cS(B)} \frac{1}{n}D_H^{\eps} \big( (\cZ_{A'\to B}(\phi_{AA'}))^{\otimes n} \big\| 1_{A^n} \otimes \sigma_{B}^{\otimes n}\big) \\
&\leq \frac{1}{n}D_H^{\eps} \big( (\cZ_{A'\to B}(\phi_{AA'}))^{\otimes n} \big\| (1_{A} \otimes \cZ_{A'\to B}(\phi_{A'}))^{\otimes n}\big)\,. \label{eq:lastone}
\end{align}
To show the second inequality we apply $\cE^{\otimes n}$ to both states and employ the data-processing inequality for the hypothesis testing relative entropy. Note that the map $\cE$ keeps $\cZ(\phi_{AA'})$ invariant and maps $1_{A} \otimes \sigma_B$ to a normalized state $\sigma_{AB}$ that is classically correlated and thus in particular in PPT$'(A:B)$. The bound in~\eqref{eq:lastone} has the form of a conditional entropy or coherent information, and it can be expanded again using~\eqref{eq:expand-hypo} to find
\begin{align}\label{eq:cov_deph}
\hat{R}^{\mathrm{cpp}}(n;\eps)\leq I(A \rangle B)_{\cZ(\phi)} + \sqrt{\frac{ V(A \rangle B)_{\cZ(\phi)}}{n}}\, \Phi^{-1}(\eps)+ O\left(\frac{\log n}{n}\right)\,.
\end{align}
So for covariant generalized dephasing channels this now agrees with the inner bound from Theorem~\ref{th:achieve-second} up to the second order asymptotically (an example being the qubit dephasing channel as discussed in Section~\ref{sec:dephasing}). Given the recent results~\cite{tomamichelww14} the outer bound~\eqref{eq:cov_deph} might also hold for generalized dephasing channels (without assuming covariance) and then agree with the inner bound from Theorem~\ref{th:achieve-second} up to the second order asymptotically.


\subsection{The Qubit Dephasing Channel}\label{sec:dephasing}

The qubit dephasing channel is defined as
\begin{align}
\cZ_{\gamma}: \rho \mapsto (1-\gamma)\rho + \gamma Z \rho Z \,,
\end{align}
where $\gamma \in [0,1]$ is a parameter and $Z$ is the Pauli $Z$ operator. This channel is covariant since it is a qubit Pauli channel. Now to determine the second order asymptotic performance it is sufficient to specialize the outer bound from~\eqref{eq:cov_deph} and to apply Theorem~\ref{th:achieve-second} for the inner bound. It is then easily seen that
\begin{align}
I(R \rangle B)_{\cZ_\gamma(\psi)} &=1-h(\gamma)\quad\mathrm{with}\quad h(\gamma)=1-\gamma \log\gamma-(1-\gamma)\log(1-\gamma)\\
V(R \rangle B)_{\cZ_\gamma(\psi)} &=v(\gamma)\quad\mathrm{with}\quad v(\gamma)=\gamma (\log \gamma + h(\gamma))^2 + (1-\gamma) (\log (1-\gamma) + h(\gamma) )^2\,,
\end{align}
and hence we deduce
\begin{align}\label{eq:2nddephasing}
\hat{R}(n;\eps) = 1-h(\gamma) + \sqrt{\frac{v(\gamma)}{n}}\, \Phi^{-1}(\eps) + O\left(\frac{\log n}{n}\right)\,.
\end{align}

However, we can refine~\eqref{eq:2nddephasing} and determine the third order asymptotic performance. We do this by directly obtaining the finite block length behavior of the qubit dephasing channel from that of the classical binary symmetric channel (BSC). First, consider the converse, particularly that of~\eqref{eq:covariant_dephasing}, applied to the channel $\cZ_\gamma^{\otimes n}$. Using the Bell states $\phi_{AB}^+=\phi_{AA'}$ and $\phi_{AB}^-=(1_{A}\otimes Z_{A'})\phi_{AA'}(1_{A}\otimes Z_{A'})$, we immediately find 
\begin{align}
\omega_{AB}:=\cN_{A'\to B}(\phi_{AA'})=(1-\gamma)\phi_{AB}^++\gamma\phi_{AB}^-\,.
\end{align}
Now, in~\eqref{eq:covariant_dephasing} we are free to pick any PPT$'$ state to obtain a bound. Pick $\sigma_{AB}=\tfrac12(\phi_{AB}^++\phi_{AB}^-)$, which gives\footnote{The choice of $\sigma_{AB}$ is equivalent to using the convex relaxation of the bound, Corollary~\ref{cor:relaxation}, and choosing $\cM=\cZ_{1/2}$ in~\eqref{eq:f_function}.}
\begin{align}
\hat{R}(n;\eps)\leq \hat{R}^{\mathrm{cpp}}(n;\eps)\leq \frac{1}{n} D_H^{\eps} \big( \omega_{AB}^{\otimes n} \big\| \sigma_{AB}^{\otimes n} \big)\,.
\end{align}
To connect to the finite block length bounds of the BSC, consider measuring both $A$ and $B$ in the Pauli $x$ basis, and let $X$ and $Y$ be the output random variables for $A$ and $B$, respectively. For the state $\omega_{AB}$, this results in the distribution $P_{XY}$ in which $P_X$ is uniformly-distributed and $P[Y=X]=1-\gamma$. For $\sigma_{AB}$, the distribution is of product form $P_XQ_Y$ with $Q_Y$ also uniform. Moreover, the original quantum states can be reconstructed from the classical random variables $X$ and $Y$ by the map which outputs $\phi_{AB}^+$ when $X=Y$ and $\phi_{AB}^-$ otherwise. Therefore, the bound becomes
\begin{align}
\hat{R}(n;\eps)\leq\frac{1}{n} D_H^{\eps} \big( P_{XY}^{\times n}\|P_X^{\times n}\times Q_Y^{\times n} \big)\,,
\end{align}
which is precisely the bound obtained by Polyanskiy {\it et al.} for the BSC~\cite[Theorem 26]{polyanskiy10}, which is equivalent to the classical sphere-packing bound~\cite[Eq.\ 5.8.19]{gallager68}. This establishes one inequality (upper bound) in~\eqref{eq:res_z}.

For the achievability, we may directly employ linear codes for the classical BSC to the qubit dephasing channel. Specifically, any linear $\{R,n,\eps\}$ code for the BSC (which recovers the input with probability at least $1-\eps$, averaged over a uniform choice of inputs), can be converted into an $\{R,n,\eps\}$ Calderbank-Shor-Steane (CSS) code for entanglement transmission over the dephasing channel.
This is possible since, for a linear code, the action of the channel is a mapping among the orthogonal Bell states, which is essentially a classical action. 

To formalize the connection, begin with the description of the classical linear code by its $(\log n-\log|M|)\times n$ parity check matrix $H$. Each row $r_j\in \{0,1\}^n$ defines a parity function and the codewords $c_k$ of the code must satisfy $c_k\cdot r_j=0$ for all $j$. 
The associated CSS code can be defined as the simultaneous $+1$ eigenspace of the ``stabilizer'' operators $X^{r_j}$, where $X^{r_j}=X^{r_{j,1}}\otimes \cdots \otimes X^{r_{j,n}}$.\footnote{
Generically, a CSS code has stabilizers of both $X$-type, as here, and of $Z$-type, i.e.\ composed of products of Pauli $Z$ operators.} Crucially, the action of the channel is to apply an operator of the form $Z^{u}$, with $u\in\{0,1\}^n$, according to the distribution $P_U$. At the output, the receiver can simultaneously determine the eigenvalues of all the of the stabilizer operators. This information is precisely equivalent to determining the value of the parity checks of the classical linear code, called the syndrome $s$. Given the syndrome, the decoder of the classical code determines a guess as to the input codeword, which is equivalent to a guess $u'(s)$ of the actual channel error. 

We may also utilize this algorithm (whatever its precise details) in the quantum case, and attempt to correct the error by applying $Z^{u'(s)}$. When $u'(s)$ is the true error pattern, the quantum state is properly recovered, and the entanglement fidelity is unity. On the other hand, if $u'(s)$ is incorrect, then in the worst case the action $Z^{u'(s)+u}$ is a logical operation on the code subspace, which results in a state orthogonal to the desired entangled state. Therefore, the error probability of the classical code translates directly into the entanglement fidelity of the quantum code. Thus, we may apply finite-block length bounds for linear codes, particularly the bound of Poltyrev~\cite{poltyrev94} (see also~\cite[Eq.\ 65]{polyanskiy10}). This establishes the other inequality (lower bound) in~\eqref{eq:res_z}.


\subsection{The Erasure Channel}

The qubit erasure channel is defined as
\begin{align}
\cE_{\beta}:\rho\mapsto(1-\beta)\rho+\beta\proj{e}\,,
\end{align}
where $\beta \in [0,1]$ is a parameter and $\proj{e}$ is a quantum state orthogonal to $\rho$. Using the covariance of the channel, we could first obtain a second order asymptotic similar to that of the dephasing channel in~\eqref{eq:2nddephasing}. However, it is not too difficult to directly derive an outer bound and an explicit coding scheme leading to an inner bound, which precisely match for all $n$.

Let us begin with the outer bound. Again we may relate the finite block length performance to a classical coding problem, namely the classical binary erasure channel (BEC). The argument for the outer bound proceeds very similarly to the dephasing example. The optimal channel input state corresponds to the maximally entangled state $\phi_{AA'}$, and the state produced by the channel is now 
\begin{align}
\omega_{AB}=(1-\beta)\phi_{AB}+\beta \pi_{A}\otimes\proj e_B\,,
\end{align}
where $\pi_{A}$ denotes the maximally-mixed state. Measurement of $A$ in the Pauli $x$ basis and $B$ in the basis $\{\ket{+},\ket{-},\ket{e}\}$ produces the distribution $P_{XY}$ with $P_X$ uniform and $Y=X$ with probability $1-\beta$ and $Y=e$ with probability $1-\beta$. The original state can be reconstructed using the map which sends $(X,Y)$ to $\phi_{AB}^+$ when $X=Y$, $\phi_{AB}^-$ when $X\neq Y\neq e$, and to $\pi_A\otimes\proj e_B$ when $Y=e$ otherwise. As before, we make a specific choice of PPT$'$ state in~\eqref{eq:covariant_dephasing}, but this time not a product state. Instead, consider the classical distribution $P_{X}^{\times n}\times Q_{Y^n}$ given in~\cite[Eq.\ 168]{polyanskiy10}. The $Q_{Y^n}$ distribution has the property that any two $y^n$ with the same number of erasure symbols $e$ have the same probability, i.e.\ there is no dependence on the number of $0$s versus $1$s. The aforementioned map takes the distribution to a quantum state which is diagonal in the standard bases $\{\ket{0},\ket{1}\}$ for $A$ and $\{\ket{0},\ket{1},\ket{e}\}$ for $B$, and is therefore a PPT state. This can be seen as follows. Consider a fixed position $j$ in a given a pair $(x^n,y^n)$. If $y_j=e$, the state of the $j$th pair of systems $AB$ is manifestly diagonal in the standard basis. On the other hand, if $y_j\neq e$, then the state is mapped to either $\phi_{AB}^+$ or $\phi_{AB}^-$ depending on the value of $x_j$. But the sequence in which $y_j$ takes the other value has identical probability, meaning the two Bell states occur with equal probability, making the $AB$ state diagonal. Since we may map $\omega_{AB}^{\otimes n}$ and $\sigma_{A^nB^n}$ to the associated classical distributions and back, the following converse holds for the qubit erasure channel,
\begin{align}\label{eq:outerbec}
\hat{R}^{\mathrm{cpp}}(n;\eps)\leq \frac{1}{n} D_H^{\eps} \big( P_{XY}^{\times n} \big\| P_X^{\times n}\times Q_{Y^n} \big)\,.
\end{align}
By design in the choice of $\sigma_{AB}$,\footnote{This corresponds to using Corollary~\ref{cor:relaxation} with $\cM$ the channel which ignores its input and prepares $\sigma_{B^n}$ at the output.} this is precisely the bound for the BEC reported by Polyanskiy {\it et al.}~\cite[Thm.~38]{polyanskiy10}, as discussed in more detail by Polyanskiy~\cite{polyanskiy13}.

Next, we construct an explicit coding scheme, involving classical post-processing including communication from the receiver to the sender, which matches the outer bound exactly. (See Figure~\ref{fig:scheme-cpp} for schematic description of this code.) The strategy of the coding scheme is to generate maximally entangled qubit states using the quantum channels and then use the successfully transmitted (i.e.\ not erased) maximally entangled qubit states to distill a an entangled state of local dimension $|M|$, as required.
(Note that the number $|M|$ is fixed at the outset of the code, i.e.\ the entanglement transmission scheme must deliver a maximally entangled state with local dimension $|M|$, possibly at the expense of low fidelity, rather than outputting a variable number of certifiably high fidelity entangled pairs.)

First, the encoder prepares $n$  maximally entangled qubit states $\ket{\psi}$ and sends one half of each over the channel. The other halves, together with the untouched system $M'$, are stored in the memory register $Q$. The decoder now works as follows. The receiver determines which qubits have not been erased and informs the sender of their locations. Let $L$ be the random variable indicating the total number of erasures and note that $L$ follows a binomial distribution with parameters $n$ and $\beta$. Let us also fix $k = \big\lceil\log|M|\big\rceil$ and consider the following two cases:

\begin{enumerate}
  \item If $L= l \leq n - k$ the decoder can extract a maximally entangled state with unit fidelity. To do so, it selects $k$ perfectly transmitted entangled qubits at the sender and receiver. Let us assume (without loss of generality) that these are in a state $\ket{\phi^+}^{\otimes k} = \frac{1}{\sqrt{2^k}} \sum_{i=1}^{2^k} |ii\rangle$.
  
  The receiver then prepares a maximally entangled state of local dimension $|M|$ by measuring the $k$ qubits with the projective  measure
  \begin{align}
    \left\{ \frac{1}{{2^k-1 \choose |M|-1}} \sum_{i \in \cS} |i\rangle\!\langle i| : \cS \subseteq \big[2^k\big] \land |\cS| = |M| \right\} \,.
  \end{align}
  The outcome, a subset $\cS$ of cardinality $|M|$, is transmitted to the sender so that both sender and receiver now share a maximally entangled state on the subspace determined by $\cS$.
  
  \item  On the other hand, if $L = l > n - k$ sender and receiver simply select the successfully transmitted qubits and embed them in a space of local dimension $|M|$. The fidelity with the target state $\ket{\phi} = \frac1{\sqrt{|M|}} \sum_{i=1}^{|M|} |ii\rangle$ is given by
\begin{align}
  F\big( |\phi^+\rangle\!\langle \phi^+|^{\otimes (n-l)} , \phi \big) = \frac{1}{|M|} \sum_{i,j=1}^{|M|} \big\langle i \big| \Big( |\phi^+\rangle\!\langle \phi^+|^{\otimes (n-l)} \Big) \big| j \big\rangle = \frac{2^{n-l}}{|M|}\,.
\end{align}
\end{enumerate}

To complete the decoding operation, the sender and receiver perform quantum teleportation to teleport $M'$ to the receiver, using the maximally entangled state prepared above as a resource. The fidelity of the state prepared above with the target state $\phi_{MM'}$ is then just the expected fidelity over $L$, which evaluates to
\begin{align}
F &= \sum_{l=0}^{n-k} {n \choose l} \beta^l(1-\beta)^{n-l}+\sum_{l=n-k+1}^{n}{n \choose l}\beta^l(1-\beta)^{n-l} \,\frac{2^{n-l}}{|M|} \\
&=1-\sum_{l=n-k+1}^n {n \choose l} \beta^l (1-\beta)^{n-l}
\left(1- \frac{2^{n-l}}{|M|} \right)\,.\label{eq:innerbec}
\end{align}
This is exactly the expression reported in the aforementioned outer bound in~\cite[Thm.~38]{polyanskiy10}, meaning the inner bound coincides with the outer bound when we allow classical post-processing and communication from the receiver to the sender.


\subsection{The Qubit Depolarizing Channel}

The qubit depolarizing channel is defined as
\begin{align}
\mathcal{D}_{\alpha}: \rho \mapsto (1-\alpha)\rho + \frac{\alpha}{3}\left(X \rho X + Y \rho Y + Z \rho Z \right)\,,
\end{align}
where $\alpha\in[0,1]$ is a parameter and $X,Y,Z$ are the Pauli operators. This channel is covariant since it is a qubit Pauli channel. Using the Bell states $\phi_{AB}^+=\phi_{AA'}$, $\phi_{AB}^-=(1_{A}\otimes Z_{A'})\phi_{AA'}(1_{A}\otimes Z_{A'})$, $\psi_{AB}^+=(1_{A}\otimes X_{A'})\phi_{AA'}(1_{A}\otimes X_{A'})$, and $\psi_{AB}^-=(1_{A}\otimes Y_{A'})\phi_{AA'}(1_{A}\otimes Y_{A'})$, we immediately find 
\begin{align}
\omega_{AB}:=(\mathcal{I}_{A}\otimes\mathcal{D}_{\alpha})(\phi_{AA'})=(1-\alpha)\phi^+_{AB}+\frac{\alpha}{3}\left(\phi^-_{AB}+\psi^+_{AB}+\psi^-_{AB}\right)\,.
\end{align}
Now choosing $\sigma_{AB}=\frac12\phi^+_{AB}+\frac16(\phi^-_{AB}+\psi^+_{AB}+\psi^-_{AB})$ in~\eqref{eq:covariant_dephasing} gives the outer bound
\begin{align}
\hat{R}_{\mathcal{D}_{\alpha}}(n;\eps)\leq\hat{R}^{\mathrm{cpp}}_{\mathcal{D}_{\alpha}}(n;\eps)\leq \frac{1}{n} D_H^{\eps} \big( \omega_{AB}^{\otimes n} \big\| \sigma_{AB}^{\otimes n} \big)\,.
\end{align}
As in the case of the qubit dephasing channel, we can convert the hypothesis test between $\omega_{AB}$ and $\sigma_{AB}$ into a test between classical distributions, in fact precisely those distributions which were used in the dephasing example. This follows by considering the map which generates $\phi^+_{AB}$ when $X=Y$ and otherwise randomly generates one of the other Bell states when $X\neq Y$. Therefore, we obtain the same outer bound for the qubit depolarization channel as for the qubit dephasing channel. This raises the question of whether cpp assistance (or more generally PPT assistance, cf.~Footnote~\ref{ft:ppt}) can turn the qubit depolarizing channel into the qubit dephasing channel.


\section{Conclusion}\label{sec:conclusion}

We gave inner (achievability) and outer (converse) bounds on the boundary of the achievable region for quantum communication with finite resources. We showed that these bounds agree for the qubit dephasing channel and qubit erasure channel with classical post-processing assistance up to the second order asymptotically. Moreover, we even gave a third order characterization for these specific examples by employing finite block length bounds of the binary symmetric channel and the binary erasure channel, respectively.

However, many questions remain open. For example, we would like to understand if the inner bound in Result~\ref{res:inner} characterizes the achievable region for all degradable channels~\cite{devetakshor05} (cf.~the open questions in~\cite{tomamichelww14}). Also it would also be interesting to explore higher order refinements for channels with zero quantum capacity (e.g., for the erasure channel with $\beta\geq1/2$ and no assistance). This might lead to a better understanding of super activation of the quantum capacity~\cite{smithyard08}.

Finally, we would like to note that the recent results about finite resource entanglement assisted classical communication~\cite{datta14} can immediately be transformed to entanglement assisted quantum communication (and this then also gives outer bounds on the achievable rate region for unassisted codes). This is accomplished by using the equivalence results in~\cite[App.~B]{matthews14} which make use of quantum teleportation and superdense coding. In particular, one finds that for covariant channels $\cN$ (which includes the qubit dephasing channel and the erasure channel) the boundary of the entanglement assisted achievable region $\hat{R}^{\textnormal{E}}(n;\eps)$ satisfies
\begin{align}
     \hat{R}^{\textnormal{E}}(n;\eps) = \frac{I(\cN)}{2} + \sqrt{\frac{V^{\eps}(\cN)}{4n}} \Phi^{-1}(\eps) + O\left(\frac{\log n}{n}\right) \,,
\end{align}
with the mutual information of the channel, $I(\cN)$, and its variance, $V^{\eps}(\cN)$, as defined in~\cite{datta14}. As an example, we mention again the qubit dephasing channel $\cZ_{\gamma}$ for which
\begin{align}\label{eq:dephazing_ent}
     \hat{R}^{\textnormal{E}}(n;\eps) = 1-2h(\gamma) + \sqrt{\frac{v(\gamma)}{4n}} \Phi^{-1}(\eps) + O\left(\frac{\log n}{n}\right) \,.
\end{align}
where $h(\gamma)$ denotes the binary entropy and $v(\gamma)$ the corresponding variance as defined in Result~\ref{res:z}.

\begin{figure}
\begin{flushleft}
\hspace{1cm}
\begin{overpic}[width=0.45\textwidth]{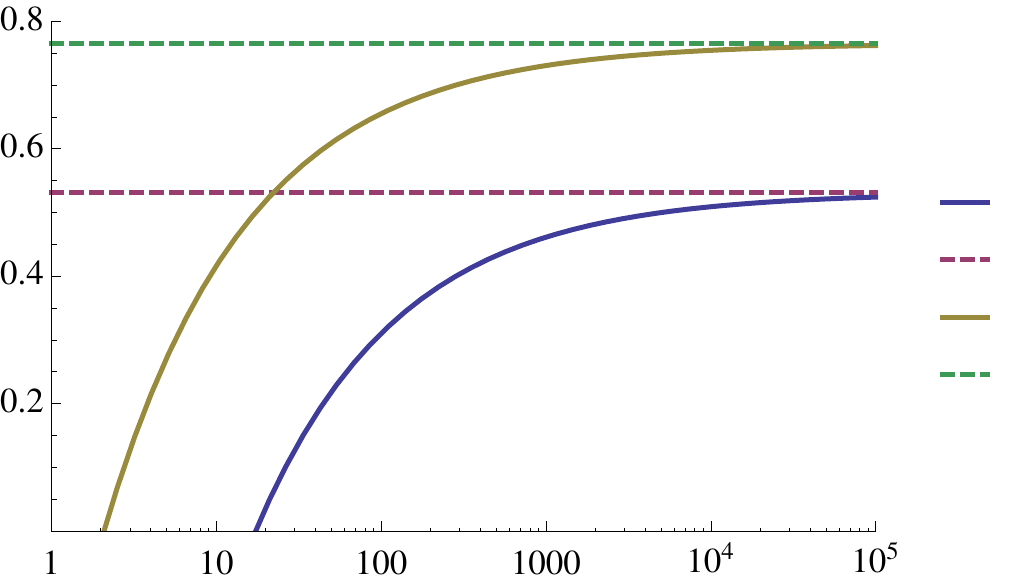}
  \put(100, 37){\footnotesize unassisted achievable region}
  \put(100, 31.5){\footnotesize unassisted capacity}
  \put(100, 26){\footnotesize entanglement-assisted achievable region}
  \put(100, 20.5){\footnotesize entanglement-assisted capacity}
  \put(30, -4){\footnotesize number of channel uses, $n$}
  \put(-5, 27){\rotatebox{90}{\footnotesize rate, $R$}}
  \put(45,27){\vector(2,-1){10}}
  \put(50,26){\footnotesize achievable region}
\end{overpic}
\end{flushleft}
 \caption{Second order approximation of the achievable region of a qubit dephasing channel with $\eps = 1\%$ and $\gamma = 0.1$ ; the achievable region is enlarged in the presence of entanglement~\cite{datta14}.}
\end{figure}
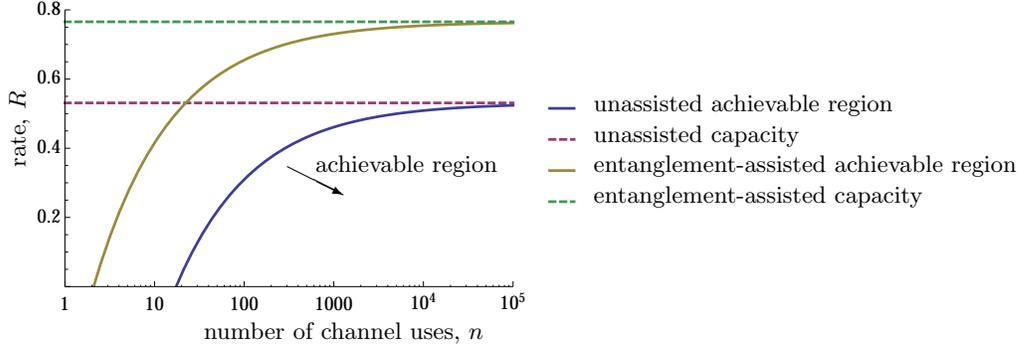


\paragraph*{Acknowledgments.} MT is funded by an University of Sydney Postdoctoral Fellowship and acknowledges support from the ARC Centre of Excellence for Engineered Quantum Systems (EQUS). 
JMR was supported by the Swiss National Science Foundation (through the National Centre of Competence in Research ‘Quantum Science and Technology’) and by the European Research Council (grant No. 258932).
We thank Chris Ferrie, Chris Granade, William Matthews, David Sutter, and Mark Wilde for helpful discussions.


\appendix


\section{Semidefinite Optimization}\label{app:sdp}

In this section we describe how to formulate the outer bound from Corollary~\ref{cor:relaxation} as a semidefinite optimization program satisfying strong duality.

A semidefinite program (SDP) is simply an optimization of a linear function of a matrix or operator over a feasible set of inputs defined by positive semidefinite constraints. We give only the bare essentials here, for more detail see \cite{boyd04,watrous-ln11}. The maximization form of an SDP is defined by a Hermiticity-preserving superoperator $\cE_{A\to B}$ taking $\cL(A)$ to $\cL(B)$, a constraint operator $C\in \cL(B)$, and an operator $K\in \cL(A)$ which defines the objective function. The SDP is the following optimization, which we will also refer to as the primal form,
\begin{align}
\begin{array}{r@{\,\,}rl}
\alpha=&\text{supremum} & \tr[KX]\\
&\text{subject to} & \cE(X)\leq C\\
&& X\geq 0\,.
\end{array}
\end{align}
When the feasible set is empty, i.e.\ no $X$ satisfy the constraints, we set $\alpha=-\infty$. The dual form arises as the optimal upper bound to the primal form, and takes the form 
\begin{align}
\begin{array}{r@{\,\,}ll}
\beta=&\text{infimum} & \tr[CY]\\
&\text{subject to} & \cE^*(Y)\geq K\\
&& Y\geq 0\,.
\end{array}
\end{align}
Again, when the set of feasible $Y$ is empty, $\beta=\infty$. Weak duality is the statement that $\alpha\leq \beta$, that indeed the dual form gives upper bounds to the primal (or that the primal lower bounds the dual). Strong duality is the statement that the optimal upper bound equals the value of the primal problem, $\alpha=\beta$. This state of affairs often holds in problems of interest, and can be established by either of the following Slater conditions. In the first, called strict primal feasibility, strong duality holds if $\beta$ is finite and there exists an $X>0$ such that $\cE(X)<C$. Contrariwise, under strict dual feasibility strong duality holds when $\alpha$ is finite and there exists a $Y>0$ such that $\cE^*(Y)>K$. 
For strongly dual SDPs we also have the so-called complementary slackness conditions $\cE^*(Y)X=KX$ and $\cE(X)Y=CY$ that relate the primal and dual optimizers.

\begin{theorem}
With the notation from Corollary~\ref{cor:relaxation}, the outer bound $f(\cN,\eps)$ can be written as
\begin{align}\label{eq:sdp_min}
\begin{array}{rcll}
f(\cN,\eps) &=& \textnormal{infimum} & \tr[\xi_A]\\
& & \textnormal{subject to} & \xi_A\otimes 1_B\geq \Lambda_{AB}+\Theta_{AB}^{T_A}\\
& & & \Lambda_{AB}\in \Gamma(\rho_A,\cN,\eps); \rho_A\in \cS(A)\\
& & & \xi_A,\Theta_{AB}\geq 0\,,
\end{array} 
\end{align}
or, equivalently,
\begin{align}\label{eq:sdp_max}
\begin{array}{rcll}
f(\cN,\eps) &=& \textnormal{supremum} & \mu(1-\eps)-\nu\\
& & \textnormal{subject to} & \mu N_{AB} \leq M_{AB}+R_{AB}\\
& & & \tr_B[R_{AB}]\leq n 1_A\\
& & & M_{AB}\in \textsc{ppt}; \mu,\nu,R_{AB}\geq 0\,.
\end{array} 
\end{align}
\end{theorem}

\begin{proof}
The proof is straightforward: we simply use the dual of the inner optimization in~\eqref{eq:f_function} to obtain the minimization problem~\eqref{eq:sdp_min}. Then we use Slater's condition to show that strong duality holds and obtain~\eqref{eq:sdp_max}.

Consider the function
\begin{align}\label{eq:f0def}
f_0(O_{AB}):=\sup_{\cM_{A\to B}\in \textsc{PPT}}\tr[O_{AB} M_{AB}]\,,
\end{align}
and observe that $f_0$ is a semidefinite program. In particular, it is a primal problem as we have defined it, with $X=M_{AB}$, $K=O_{AB}$, $C=(0,1_A)$, and $\cE(X)=(-X^{T_A},\tr_B[X])$.
Choosing for the dual variables $Y=(\Gamma_{AB},\xi_A)$, the dual of $f_0$ is 
\begin{align}
\begin{array}{crl}
\tilde{f}_0(O_{AB}):= &\textnormal{infimum} &\tr[\xi_A]\\
& \textnormal{subject to} & \xi_A\otimes1_B\geq O_{AB}+\Gamma_{AB}^{T_A}\\
&& \Gamma_{AB},\xi_A\geq 0\,.
\end{array}
\end{align}
Combining this with the outer optimization in~\eqref{eq:f_function} gives the minimization program~\eqref{eq:sdp_min}. The equality statement is precisely strong duality of the primal and dual forms of the inner optimization. By Slater's condition, strong duality holds if $f_0$ is finite and there exists a strictly feasible set of dual variables. Observe that $f_0(O_{AB})\leq |A|$, since for the optimal $M_{AB}$ we have $f_0(O_{AB})=\tr[M_{AB}O_{AB}]\leq \tr[M_{AB}]\leq \tr_A[1_A]=|A|$. Here we have used the upper bounds $O_{AB}\leq 1_{AB}$ and $\tr_B[M_{AB}]\leq1_A$. Thus, the first condition is fulfilled. Meanwhile, $\Gamma_{AB}=1_{AB}$ and $\xi_A=3\cdot1_A$ are a strictly feasible pair. Thus, $\tilde f_0=f_0$ over the domain of interest. 

To construct the maximization program, we simply dualize the minimization program. In particular, $f(\cN,\eps)$ is a dual-form semidefinite program in the variable $Y=(\phi_A,\Lambda_{AB},\Gamma_{AB},\xi_A)$ with $C=(0,0,0,1_A)$, $K=(1-\eps,-1,0,0)$, and 
\begin{align}
\cE^*(Y)=(\tr[N_{AB}\Lambda_{AB}],-\tr[\phi_A],\phi_A^T\otimes1_B-\Lambda_{AB},\xi_A\otimes1_B-\Lambda_{AB}-\Gamma^{T_A}_{AB})\,.
\end{align}
Choosing primal variables $X=(m,n,R_{AB},M_{AB})$ leads to the maximization in~\eqref{eq:sdp_max}. Equality again follows from Slater's condition: $f$ is finite (in particular the bound on $f_0$ used above), while a feasible choice of dual variables is given by $M_{AB}=R_{AB}=\frac 1{2|B|}1_{AB}$, $n=1$, and $m=\frac 1{2|A||B|}$. The choice of $m$ ensures the first constraint holds strictly, since any Choi operator of a trace-preserving map satisfies $\| N_{AB}\|_\infty=|A|$ (largest singular value).
\end{proof}

No discussion of strong duality of semidefinite programs is complete until the complementary slackness conditions have been formulated. Often, these give considerable insight into the form and properties of the optimizing variables. First observe that 
\begin{align}
\cE(X)=(-n1_A+\tr_B[R_{AB}^{T_A}]
,\,mN_{AB}-M_{AB}-R_{AB}
,\,-M_{AB}^{T_A},\,\tr_B[M_{AB}])\,.
\end{align}
Then the conditions are easy to read off from the form of $C$ and $K$. They are 
\begin{align}
\tr[\phi_A]&=1\\
\tr[\Lambda_{AB}N_{AB}]&=1-\eps\\
\phi_A^T R_{AB}&=\Lambda_{AB}R_{AB}\\
\xi_AM_{AB}&=(\Lambda_{AB}+\Gamma^{T_A}_{AB})M_{AB}\\
n\phi_A&=\tr_B[R_{AB}^{T_A}
]\phi_A\\
M_{AB}^{T_A}\Gamma_{AB}&=0\\
\tr_B[M_{AB}]\xi_A&=\xi_A\\
m N_{AB}\Lambda_{AB}&=(M_{AB}+R_{AB}
)\Lambda_{AB}\,.
\end{align}

\bibliographystyle{arxiv_no_month}
\bibliography{library,jmrbib}

\end{document}